\documentclass[11pt,a4paper,english]{amsart}
\usepackage[utf8]{inputenc}
\usepackage[english]{babel}

\usepackage{amsmath}
\usepackage{amsthm}
\usepackage{graphicx}
\usepackage{xcolor}
\usepackage{amsfonts}
\usepackage[biblabel]{cite}
\usepackage{bm}
\usepackage[hidelinks]{hyperref}
\usepackage{amssymb}
\usepackage{amscd}
\usepackage{etoolbox}
\patchcmd{\thmhead}{(#3)}{#3}{}{}
\usepackage{enumitem}
\usepackage{breqn}
\usepackage[nameinlink,capitalise,noabbrev]{cleveref}
\usepackage{braket}

\usepackage{mathtools}
\DeclarePairedDelimiter\abs{\lvert}{\rvert}
\DeclarePairedDelimiter\norm{\lVert}{\rVert}
\makeatletter
\let\oldabs\abs
\def\abs{\@ifstar{\oldabs}{\oldabs*}}
\let\oldnorm\norm
\def\norm{\@ifstar{\oldnorm}{\oldnorm*}}
\makeatother

\newtheorem{theorem}{Theorem}[section]
\newtheorem{proposition}[theorem]{Proposition}
\newtheorem{corollary}[theorem]{Corollary}
\newtheorem{lemma}[theorem]{Lemma}
\theoremstyle{definition}
\newtheorem{definition}[theorem]{Definition}
\newtheorem{remark}[theorem]{Remark}
\newtheorem{example}[theorem]{Example}

\DeclareMathOperator{\Rep}{Rep}

\newcommand{\vladut}{Vl\u{a}du\c{t}}
\newcommand{\F}{{\mathbb{F}}}

\newcommand{\Tr}{\operatorname{Tr}}
\newcommand{\wt}{\operatorname{wt}}
\newcommand{\supp}{\operatorname{supp}}
\newcommand{\Span}{\operatorname{span}}

\newcommand{\diag}{\operatorname{diag}}
\newcommand{\cC}{\mathcal{C}}
\newcommand{\cA}{\mathcal{A}}
\newcommand{\cQ}{\mathcal{Q}}

\newcommand{\one}{\mathbf{1}}
\newcommand{\size}[1]{\left|#1\right|}

\title[Asymptotically good binary triorthogonal codes]{Asymptotically good binary triorthogonal codes and higher-level transversal gates}
\author{Rodrigo San-José}
\address[Rodrigo San-José]{Department of Mathematics\\ Virginia Tech\\ Blacksburg, VA USA}
\email{rsanjose@vt.edu}

\thanks{The author was partially supported by the NSF grant DMS-2401558, the Commonwealth Cyber Initiative, an AMS-Simons Travel Grant, and by Grant PID2022-138906NB-C21 funded by MICIU/AEI/10.13039/501100011033 and by ERDF/EU}
\subjclass[2020]{Primary: 94B05. Secondary: 81P70, 94B27, 81P65}
\keywords{CSS codes, triorthogonal, CSS-T, non-Clifford, algebraic geometry codes}

\begin{document}

\begin{abstract}
For each level of the Clifford hierarchy from the third level onward, we construct explicit asymptotically good binary CSS codes for which a strongly transversal diagonal rotation, followed by a correction from the preceding level, implements the corresponding logical rotation. We achieve this using algebraic geometry codes over binary extension fields and performing an alphabet reduction that translates orthogonality conditions over the extension field to binary overlap conditions. In particular, we obtain the first asymptotically good family of binary triorthogonal codes. We also give a stronger construction with worse parameters for which the strongly transversal physical rotation implements the desired logical rotation, with no subsequent correction. Finally, adapting an error-correction-based distillation protocol, our binary $T$-gate families yield direct constant-overhead $T$-state block distillation from noisy $T$-state inputs under sufficiently weak input noise and ideal stabilizer operations.
\end{abstract}

\maketitle

\section{Introduction}
Transversal gates are among the simplest fault-tolerant logical operations, because a fault in one physical gate cannot spread within a code block. We focus on the single-qubit non-Clifford diagonal rotations, whose first case is the \(T\) gate. A central distinction in this subject is between preserving the code space and inducing the desired non-Clifford logical action. CSS-T codes, introduced in \cite{calderbankclassicalcsst,rengaswamyOptimalityCSST}, are CSS codes that are preserved under the action of a transversal $T$ gate. While they are interesting in their own right and have been extensively studied recently \cite{sanjose_CSST,ruano_csst_cartesian,kashyap_asymptotically_good_csst,sanjose_allerton_triorthogonal}, the logical action they implement need not be non-Clifford; e.g., see \cite{albert_T_gate,kashyap_asymptotically_good_csst}. For example, \cite{berardini_asymptotically_csst} constructs asymptotically good binary CSS-T codes whose logical action is the identity.

Triorthogonal codes, introduced in \cite{bravyiTriorthogonalOriginal}, are a class of CSS codes which, up to a Clifford correction, implement a logical $T$ gate when one applies a transversal $T$ gate on the physical qubits. This is crucial if we want to achieve universal quantum fault-tolerant computation, since the Clifford group, together with the $T$ gate, forms a universal gate set. One way to perform the $T$ gate is via state injection, which requires magic states produced via magic state distillation. The original paper \cite{bravyiTriorthogonalOriginal} already explains how to perform magic state distillation starting from a triorthogonal code. Thus, constructing good triorthogonal codes in the finite-length regime is a topic of considerable interest \cite{albert_T_gate,rosnes_triorthogonal}. Moreover, the distillation overhead is related to the scaling exponent $\gamma=\log(N/K)/\log d_Z$. Finding codes with a small scaling exponent, or families for which $\gamma$ tends to zero \cite{wills_csst_decreasing_improved}, is also an active research area. We stress that $\gamma\to0$ in this sense does not by itself imply constant-overhead magic state distillation. Recently, it was shown in \cite{wills_constant_overhead_msd} that constant overhead can be achieved by combining asymptotically good triorthogonal codes over binary extension fields with a new single-round distillation protocol based on error correction rather than post-selection.

Constructions in which $\gamma$ tends to zero are already known \cite{wills_csst_decreasing_improved}, and asymptotically good triorthogonal codes over binary extension fields have recently been used to obtain constant-overhead magic state distillation \cite{wills_constant_overhead_msd}. Still, it remains interesting to find asymptotically good \emph{binary} triorthogonal codes. Indeed, the codes in \cite{wills_constant_overhead_msd} are naturally defined over a binary field extension of degree $m$; after identifying a qudit with $m$ qubits, their transversal gate becomes an $m$-qubit block gate rather than single-qubit transversal $T^{\otimes N}$ gates. 

In \cite{guruswami_good_transversal}, the authors construct asymptotically good families of codes with a transversal $CCZ$. For prime dimensions $q\geq 5$, their codes also have a transversal diagonal non-Clifford gate which may be regarded as a qudit generalization of the $T$ gate. Their construction also works for qubits, but the corresponding single-qubit diagonal gate is Clifford. A simultaneous result \cite{nguyen_good_css_ccz} gives asymptotically good binary codes with transversal $CCZ$. In \cite{wills_asymptotically_good_non_clifford}, the authors further obtain asymptotically good binary codes with transversally addressable $CCZ$, allowing a logical $CCZ$ to be applied to any chosen triple of logical qubits across one or several code blocks. Their related framework of addressable orthogonality also extends formally to transversal, addressable $T$ gates up to Clifford corrections, although no asymptotically good instantiation for $T$ is obtained \cite{wills_addressable_non_clifford} (in particular, see their Open Problems 3 and 4). Obtaining asymptotically good constructions for the single-qubit $T$ gate, with the desired logical action $T$, has remained elusive; even the recent constructions in \cite{wills_csst_decreasing_improved} only achieve a constant rate with sub-polynomial distance. Very recent work also highlights the existence of asymptotically good CSS codes with a strongly transversal \(T\) gate as an open problem \cite{kashyap_z_rotations}.

This paper closes that gap: we construct the first asymptotically good family of binary triorthogonal codes, showing that linear rate and linear distance can be achieved simultaneously. We use explicit algebraic geometry (AG) codes over binary extension fields and an alphabet-reduction map that preserves orthogonality conditions. Moreover, the general constructions that we give also allow us to obtain asymptotically good CSS codes with transversal diagonal gates at any fixed level $h\geq 3$ of the Clifford hierarchy. We also give a second construction, with worse parameters, in which no subsequent correction is required. Adapting the error-correction-based protocol of \cite{wills_constant_overhead_msd}, our binary triorthogonal families give a direct constant-overhead \(T\)-state block-distillation protocol, using noisy \(T\)-state inputs directly rather than intermediate multi-qubit magic-state conversions. Here and throughout, explicit means that we can construct the generator matrices of the codes deterministically from the specified tower level, and we can construct the concatenation map.

The construction proceeds as follows. We consider asymptotically good AG codes over a binary extension such that a high enough Schur power of the code is contained in its dual. We then shorten and puncture these codes on a suitable set of coordinates, and use an additive concatenation map to pass from the extension field to binary codes. The concatenation map, which we prove always exists, is chosen so that the relevant Schur-orthogonality identities of the original AG code are transferred to overlap identities for its binary image. After puncturing, these identities express the overlaps of the resulting binary rows in terms of the deleted coordinates. By choosing suitable rows associated with these coordinates, together with the binary image of the shortened code, we obtain binary codes \(C_2\subset C_1\) with the weight and overlap conditions required for a transversal diagonal gate to induce the desired logical action. The distance of the resulting CSS codes is controlled by the distances of the original AG code and its dual. The main difference with previous constructions, such as the ones in \cite{guruswami_good_transversal,nguyen_good_css_ccz}, which use multiplication-friendly embeddings, is that our alphabet-reduction maps are designed to translate the higher-order Schur-orthogonality identities into the binary overlap conditions needed for single-qubit diagonal gates.

With respect to the structure of the paper, in \cref{s:preliminaries} we introduce the basic tools that we need for the rest of the paper: AG codes, CSS codes, and transversal gates, and lemmas about the phases induced by diagonal gates. In \cref{sec:outer}, we describe the family of AG codes that we consider, their parameters, and their orthogonality properties. In \cref{ss:map}, we introduce the maps we will use to reduce to the binary alphabet. In \cref{s:concatenate}, we construct asymptotically good CSS codes for which a strongly transversal diagonal rotation, followed by a correction from the preceding level of the Clifford hierarchy, implements the desired logical rotation at every fixed level $h\geq 3$. To complement this result, in \cref{sec:correction-free}, we give another construction of asymptotically good CSS codes with diagonal strongly transversal gates, but in this case we do not require a correction, at the cost of producing worse parameters. Finally, in \cref{sec:families}, we give some particular examples of the concatenation maps we can use for the case of the $T$ gate, obtaining explicit constructions of asymptotically good triorthogonal codes. Two of them require a Clifford correction, while the last one does not require any correction. Adapting the arguments from \cite{wills_constant_overhead_msd}, we also show that the $T$-gate families yield constant-overhead magic-state distillation in \cref{ss:constant-overhead}.

\section{Preliminaries}\label{s:preliminaries}
Let $q=2^m$ be a power of $2$, and let $\Tr=\Tr_{\F_{2^m}/\F_2}$ denote the absolute trace. A \textit{linear code} $C\subset \F_q^n$ is a linear subspace of $\F_q^n$. Given $c\in C$, we denote by $\wt(c):=|\{ i:c_i\neq 0 \}|$ the \textit{Hamming weight} of $c$. The dimension of $C$ is its dimension as a linear subspace, and its \textit{minimum distance} is $d:=\min \{\wt(c):c\in C,\;c\neq 0\}$. The code $C^\perp$ denotes the Euclidean dual of $C$, and we denote $\one:=(1,\dots,1)$. 

\subsection{Schur products and AG codes}\label{ss:ag_codes}

For vectors over a field, $x\star y$ denotes coordinate-wise multiplication. If $C$ is a linear code, $C^{\star t}$ is the span of all $t$-fold Schur products of codewords of $C$. This is an operation that appears naturally in the context of CSS-T and triorthogonal codes \cite{rengaswamyOptimalityCSST,sanjose_CSST,sanjose_transversal_css}.

Let $\mathcal{F}/\F_{2^m}$ be an algebraic function field of genus $g$. A divisor of $\mathcal{F}$ is a finite formal sum $G=\sum_P n_P P$, with $n_P\in \mathbb{Z}$, over the places of $\mathcal{F}$. Its support is $\supp(G):=\{P:n_P\neq 0\}$ and its degree is $\deg G= \sum_P n_P\deg P$. We say that $G$ is \textit{effective} if $n_P\geq 0$ for all $P$, and we write $G\geq 0$.  Given a function $f$ or a differential $\omega$, let $(f)$ and $(\omega)$ denote their associated divisors. 

Let $D=P_1+\cdots+P_n$ be a sum of distinct rational places, and let $G$ be a divisor whose support is disjoint from that of $D$. The Riemann--Roch space associated to the divisor $G$ is then $L(G):=\{f \in \mathcal{F}: (f)+G\geq 0\}\cup \{0\}$. The AG code associated to $D$ and $G$ is
\[
 C_L(D,G):=\{(f(P_1),\ldots,f(P_n)):f\in L(G)\}\subset \F_q^n.
\]
When $2g-1\leq\deg G<n$, the standard estimates for their parameters are
\begin{equation}\label{eq:ag-parameters}
 \dim C_L(D,G)=\deg G+1-g,
 \; d\bigl(C_L(D,G)\bigr)\geq n-\deg G.
\end{equation}
Furthermore,
\begin{equation}\label{eq:star-product}
 C_L(D,G)^{\star t}\subset C_L(D,tG).
\end{equation}
Choose a differential $\omega$ with a simple pole and residue one at each $P_a$, and put $R=(\omega)$. Then
\begin{equation}\label{eq:ag-dual}
 C_L(D,G)^\perp=C_L(D,R+D-G),
 \;
 d\bigl(C_L(D,G)^\perp\bigr)\geq\deg G+2-2g.
\end{equation}
These facts are standard; see \cite[Chs.~2 and 8]{stichtenoth_book}.

\subsection{CSS codes, gates, and triorthogonal codes}
For binary codes $C_2\subset C_1\subset\F_2^N$, we denote the corresponding CSS code by $\operatorname{CSS}(C_1,C_2)$. If $g^{(1)},\ldots,g^{(K)}$ represent a basis of $C_1/C_2$, the corresponding logical computational-basis states are
\[
 |\bar u\rangle
 =\frac1{\sqrt{|C_2|}}
 \sum_{c\in C_2}
 \left|c+\sum_{a=1}^K u_ag^{(a)}\right\rangle,
 \; u\in\F_2^K.
\]
Thus, it encodes $K:=\dim C_1-\dim C_2$ qubits, and we have
\[
 d_X=\min\{\wt(x):x\in C_1\setminus C_2\},\;
 d_Z=\min\{\wt(z):z\in C_2^\perp\setminus C_1^\perp\}.
\]
The parameters of the corresponding quantum code are given by $[[N,K,\Delta]]_2$, where $\Delta=\min\{d_X,d_Z\}$. For a CSS code $\cQ$, in general, we denote by $d_X(\cQ)$ and $d_Z(\cQ)$ its $X$ and $Z$ minimum distances, respectively, and $\Delta:=\min\{d_X(\cQ),d_Z(\cQ)\} $.

\begin{definition}
We say that a family of quantum codes $\{\cQ_i\}_i$ with parameters $[[N_i,K_i,\Delta_i]]_2$ is \textit{asymptotically good} if 
$$
\lim_{i\to \infty}N_i=\infty,\;\liminf_{i\to \infty }\frac{K_i}{N_i}>0 \text{ and }\liminf_{i\to \infty }\frac{\Delta_i}{N_i}>0.
$$ 
\end{definition}

Let $\mathfrak C_h^{(n)}$ denote the $h$th level of the $n$-qubit Clifford hierarchy:
$\mathfrak C_1^{(n)}$ is the $n$-qubit Pauli group $\mathcal{P}_n$, and, recursively,
\[
\mathfrak C_{h+1}^{(n)}
=
\left\{
U\in U(2^n):
UPU^\dagger\in\mathfrak C_h^{(n)}
\text{ for every }P\in\mathcal P_n
\right\}.
\]
In particular, $\mathfrak C_2^{(n)}$ is the $n$-qubit Clifford group. When the number of qubits is clear from context, we simply write $\mathfrak{C}_h$. We focus mainly on diagonal gates. For $h\geq3$, we consider $\zeta_{2^h}=e^{2\pi i/2^h}$ and
\[
 R_h:=\diag(1,\zeta_{2^h})=\begin{pmatrix}
     1&0\\
     0&\zeta_{2^h}
 \end{pmatrix}.
\]
Then $R_h\in\mathfrak C_h$, and the $T$ gate is the particular case $R_3$. If a gate is of the form $U_1\otimes \cdots \otimes U_N$, then we say it is \textit{transversal}, and if $U_1=\cdots=U_N$, we say it is \textit{strongly transversal}. Strongly transversal diagonal gates implementing a non-identity logical gate on an error-detecting CSS code are restricted to powers of gates of this form \cite[Thm.~1]{anderson_classification_transversal_gates}.

A binary matrix is \textit{triorthogonal} if the coordinate-wise product of any two distinct rows and of any three distinct rows has even weight. If its first $K$ rows have odd weight and all remaining rows have even weight, the span of all rows and the span of the even-weight rows define a CSS code encoding $K$ qubits. In \cite{bravyiTriorthogonalOriginal}, the authors prove that a diagonal Clifford correction turns physical $T^{\otimes N}$ into logical $T^{\otimes K}$. We give a general form of the correction in \cref{lem:higher-correction}. More generally, we consider the following definition.

\begin{definition}\label{def:binary-h-orthogonality}
Let $G$ be a binary matrix whose first $K$ rows have odd weight and whose remaining rows have even weight. We call $G$ \textit{$h$-orthogonal} if
\[
 \wt\left(\mathop{\bigstar}_{a\in S}g^{(a)}\right)
 \equiv0\bmod 2,
\]
for every set $S$ of row indices with $2\leq |S|\leq h$. This is the binary, unweighted specialization of the definition in \cite[Sec.~4.1]{wills_addressable_non_clifford}. For $h=3$ it recovers triorthogonality.
\end{definition}

\subsection{Diagonal phases}

For $\nu=(\nu_1,\ldots,\nu_N)\in\mathbb Z^N$ and
$x\in\F_2^N$, write
\[
 \wt_\nu(x):=\sum_{j=1}^N\nu_jx_j.
\]
For the $T$ gate, we denote $T_\nu:=\bigotimes_{j=1}^N T^{\nu_j}$. We use the following standard weight-expansion identity; see, for example, \cite[Lem. III.2]{haah_towers_divisible_codes}.

\begin{lemma}\label{lem:weight-expansion}
Let $g^{(1)},\ldots,g^{(M)}\in\F_2^N$. For every
$u=(u_1,\ldots,u_M)\in\F_2^M$ and every $\nu\in\mathbb Z^N$, we have
\begin{equation*}\label{eq:weight-expansion}
 \wt_\nu\left(\sum_{a=1}^M u_ag^{(a)}\right)
 =
 \sum_{\varnothing\neq S\subset\{1,\ldots,M\}}
 (-2)^{|S|-1}
 \left(\prod_{a\in S}u_a\right)
 \wt_\nu\left(\mathop{\bigstar}_{a\in S}g^{(a)}\right).
\end{equation*}
In particular, modulo $2^h$, only the terms with $|S|\leq h$ contribute.
\end{lemma}

\begin{proof}
For $z_1,\ldots,z_M\in\F_2$, regarded as integers, we have
\[
 (z_1+\cdots+z_M)\bmod2
 =\sum_{\varnothing\neq S\subset\{1,\ldots,M\}}
 (-2)^{|S|-1}\prod_{a\in S}z_a.
\]
Apply this identity at each coordinate with
$z_a=u_ag_j^{(a)}$, multiply by $\nu_j$, and sum over $j$.
\end{proof}

The next result is the higher-level analog of the Bravyi--Haah phase argument. Similar arguments already appear in the literature in the context of quasitransversal gates; see \cite[Prop.~9, App.~B, and Sec.~III-C]{vuillot_quantum_pin_codes} and \cite[Def.~4 and Lem.~1]{campbell_msd_synthesis}. We include the proof to specify the correction in our notation.

\begin{lemma}\label{lem:higher-correction}
Let $h\geq3$, and let $G\in\F_2^{M\times N}$ be a full-rank $h$-orthogonal binary matrix whose first $K$ rows have odd weight. Let $C_1$ be the row space of $G$, and let $C_2$ be the span of its even-weight rows. Then there is a computable diagonal correction $U\in\mathfrak C_{h-1}$ such that
$$
 UR_h^{\otimes N}\big|_{\operatorname{CSS}(C_1,C_2)}
 =\overline R_h^{\otimes K}.
$$
\end{lemma}

\begin{proof}
Let $\epsilon_a=1$ for $a\leq K$ and $\epsilon_a=0$ otherwise, and consider
\[
 \wt(g^{(a)})=\epsilon_a+2\lambda_a,
 \;
 \wt\left(\mathop{\bigstar}_{a\in S}g^{(a)}\right)=2\lambda_S, \;
 2\leq|S|\leq h.
\]
By \cref{lem:weight-expansion}, there is a polynomial $F:\F_2^M\to \mathbb Z/2^{h-1}\mathbb Z$ such that
\begin{equation}\label{eq:higher-correction-phase}
 \wt\left(\sum_{a=1}^M u_ag^{(a)}\right)
 \equiv\sum_{a=1}^K u_a+2F(u)\bmod2^h.
\end{equation}
More precisely, we may take
\begin{equation*}\label{eq:higher-correction-polynomial}
 F(u)
 =\sum_{a=1}^M\lambda_au_a
 +\sum_{\substack{S\subset\{1,\ldots,M\}\\2\leq|S|\leq h-1}}
 (-2)^{|S|-1}\lambda_S\prod_{a\in S}u_a
 \bmod2^{h-1}.
\end{equation*}
The terms with $|S|=h$ vanish modulo $2^{h-1}$. In particular, the coefficient of every monomial in $F$ of degree $r\geq1$ is divisible by $2^{r-1}$.

Let $G$ be the matrix with rows $g^{(a)}$, and choose
$B\in\F_2^{M\times N}$ such that $BG^{\mathsf T}=I_M$. For
$x\in\F_2^N$, let $L_a(x)$ be the $a$th coordinate of $Bx^{\mathsf T}$,
and define the diagonal operator
\[
 U|x\rangle =\zeta_{2^{h-1}}^{-F(L_1(x),\ldots,L_M(x))}|x\rangle=\zeta_{2^{h}}^{-2F(L_1(x),\ldots,L_M(x))}|x\rangle.
\]

We next check that $U\in\mathfrak C_{h-1}$. By \cref{lem:weight-expansion}, the integer polynomial representing each $L_a$ has its degree-$d$ coefficients divisible by $2^{d-1}$. Consider a monomial of degree $r$ of $F$. After the substitution $F(L_1(x),\dots,L_M(x))$, take one of the terms arising from this monomial, with degrees $d_1,\dots,d_r$ from the corresponding $L_a$. Its coefficient is divisible by
\[
 2^{r-1}\prod_{a=1}^r2^{d_a-1}
 =2^{d_1+\cdots+d_r-1}.
\]
Replacing powers $x_p^b$ by $x_p$ can only decrease the degree.
Hence the coefficient of every degree-$d$ monomial in
$F(L_1(x),\ldots,L_M(x))$ is divisible by $2^{d-1}$.
The claim follows from
\cite[Thm.~3]{gottesman_diagonal_clifford_hierarchy}.

For $u\in\F_2^M$, put $x=\sum_{a=1}^M u_ag^{(a)}\in C_1$. By \cref{eq:higher-correction-phase}, we have
$$
R_h^{\otimes N}\ket{x}=\zeta_{2^{h}}^{\wt(x)}|x\rangle=\zeta_{2^{h}}^{\sum_{a=1}^K u_a+2F(u)}|x\rangle.
$$
Since $L_a(\sum_b u_b g^{(b)})=u_a$, we obtain
$$
UR_h^{\otimes N}\ket{x}=\zeta_{2^{h}}^{\sum_{a=1}^K u_a}|x\rangle.
$$
Therefore, for any logical basis state $\ket{\bar u }$, we get
\[
\begin{aligned}
UR_h^{\otimes N}|\bar u\rangle
=
\frac1{\sqrt{|C_2|}}
\sum_{c\in C_2}
UR_h^{\otimes N}
\left|c+\sum_{a=1}^K u_ag^{(a)}\right\rangle&=
\zeta_{2^h}^{u_1+\cdots+u_K}
\frac1{\sqrt{|C_2|}}
\sum_{c\in C_2}
\left|c+\sum_{a=1}^K u_ag^{(a)}\right\rangle\\
&=
\zeta_{2^h}^{u_1+\cdots+u_K}|\bar u\rangle.
\end{aligned}
\]
\end{proof}

\begin{example}
For $h=3$, the correction from the previous result is a diagonal Clifford correction, while for $h=4$ it belongs to the third level of the Clifford hierarchy.
The fact that we have to consider conditions modulo increasing powers of $2$ was already observed in \cite{sanjose_transversal_css}. 
\end{example}

\section{Asymptotically good AG codes over binary extension fields}\label{sec:outer}
We use the Galois tower described in \cite[Thm.~7.4.15, Cor.~7.4.16, and the proof of Thm.~8.4.9]{stichtenoth_book}, which is obtained from the Garcia--Stichtenoth construction \cite{garcia_stichtenoth_tower}; see also \cite[Thm.~3.6 and Rem.~3.8]{nguyen_good_css_ccz} and \cite{wills_asymptotically_good_non_clifford}. In concatenation terms, these correspond to the \textit{outer codes}. Let $q=\ell^2$, $\ell=2^r\geq8$, so that $m=2r$ and $\F_{2^m}=\F_q$ is a binary extension field of even degree. There is an explicit sequence of function fields $\mathcal{F}_i/\F_{2^m}$ with $n_i\to \infty$ rational points, genus $g_i$, effective divisors $A_i,B_i$, and increasing even integers $\alpha_i,\beta_i \geq 2$ such that
\begin{align}
 g_i&=1+\frac{n_i}{\ell-1}
 \left(1-\frac1{\alpha_i}-\frac1{\beta_i}\right),
 \label{eq:tower-genus}\\
 \alpha_i\deg A_i&=\beta_i\deg B_i=\frac{n_i}{\ell-1},
 \label{eq:tower-degrees}\\
 R_i&=(\ell\alpha_i-2)A_i+(\beta_i-2)B_i-D_i.
 \label{eq:tower-canonical}
\end{align}
Here $D_i$ is the sum of the evaluation places, $R_i$ is a suitable canonical divisor, and the supports of $A_i,B_i$ are disjoint from that of $D_i$. More precisely, in the notation of \cite[Proof of Thm.~8.4.9]{stichtenoth_book}, we have $D_i=(u_0)_0$ and $R_i=(du_0/u_0)$. Since the zeros of $u_0$ in $D_i$ are simple, the differential $du_0/u_0$ has a simple pole and residue one at every evaluation place. Thus this choice of $R_i$ satisfies \cref{eq:ag-dual}.

Let $\tau$ be a power of $2$ such that $4\leq\tau\leq\ell/2$. Set
\begin{equation}\label{eq:G-choice-general}
 a_i^{(\tau)}=\frac{\ell}{\tau}\alpha_i-1,
 \;
 b_i^{(\tau)}=\left\lfloor\frac{\beta_i-2}{\tau}\right\rfloor,
 \;
 G_i^{(\tau)}=a_i^{(\tau)}A_i+b_i^{(\tau)}B_i,
\end{equation}
and let
\[
 \cC_i^{(\tau)}:=C_L(D_i,G_i^{(\tau)}).
\]

\begin{remark}\label{r:bounds_bi}
Note that $b_i^{(\tau)}\leq \beta_i/\tau$, and, since $\beta_i$ is even and $\tau$ is a power of $2$, we also have $b_i^{(\tau)}+1\geq \beta_i/\tau$.
\end{remark}

\begin{proposition}\label{prop:outer}
For every $i$, the code $\cC_i^{(\tau)}$ satisfies
\begin{equation*}\label{eq:outer-schur-condition}
 \bigl(\cC_i^{(\tau)}\bigr)^{\star(\tau-1)}
 \subset\bigl(\cC_i^{(\tau)}\bigr)^\perp, \text{ or, equivalently, }
 \sum_{j=1}^{n_i}c_j^{(1)}\cdots c_j^{(\tau)}=0
\end{equation*}
for all $c^{(1)},\ldots,c^{(\tau)}\in\cC_i^{(\tau)}$. Moreover,
\begin{align}
 \frac{\dim_{\F_{2^m}}\cC_i^{(\tau)}}{n_i}
 &\geq\rho_\ell^{(\tau)}
 :=\frac{\ell-(\tau-1)}{\tau(\ell-1)},
 \notag \\
 \frac{d(\cC_i^{(\tau)})}{n_i}
 &\geq A_\ell^{(\tau)}
 :=\frac{(\tau-1)\ell-(\tau+1)}{\tau(\ell-1)},
 \notag\\
 \frac{d((\cC_i^{(\tau)})^\perp)}{n_i}
 &\geq B_\ell^{(\tau)}
 :=\frac{\ell-(2\tau-1)}{\tau(\ell-1)}.
 \label{eq:outer-dual-general}
\end{align}
\end{proposition}

\begin{proof}
The choices from \cref{eq:G-choice-general} give
\[
 \tau G_i^{(\tau)}
 \leq(\ell\alpha_i-\tau)A_i+(\beta_i-2)B_i
 \leq (\ell\alpha_i-2)A_i+(\beta_i-2)B_i = R_i+D_i.
\]
By \cref{eq:star-product,eq:ag-dual},
\[
 \bigl(\cC_i^{(\tau)}\bigr)^{\star(\tau-1)}
 \subset C_L(D_i,(\tau-1)G_i^{(\tau)})
 \subset C_L(D_i,R_i+D_i-G_i^{(\tau)})
 =\bigl(\cC_i^{(\tau)}\bigr)^\perp.
\]
Using \cref{r:bounds_bi,eq:tower-degrees}, we get
\[
 \frac{\deg G_i^{(\tau)}}{n_i}
 =\frac1{\ell-1}\left(
 \frac\ell\tau-\frac1{\alpha_i}
 +\frac{b_i^{(\tau)}}{\beta_i}\right)
 \leq\frac{\ell+1}{\tau(\ell-1)}\leq \frac{1}{4}\left(1+\frac{2}{7}\right)<1,
\]
where we have used the fact that $4\leq \tau$ and $\ell\geq 2 \tau$. This implies that the evaluation map defining $\cC_i^{(\tau)}$ is injective. On the other hand, using \cref{eq:tower-genus,r:bounds_bi}, we have
\begin{align*}
 \frac{\deg G_i^{(\tau)}-(2g_i-2)}{n_i}
 &=\frac1{\ell-1}\left(
 \frac\ell\tau-2+\frac1{\alpha_i}
 +\frac{b_i^{(\tau)}+2}{\beta_i}\right)\\
 &\geq\frac{\ell-(2\tau-1)}{\tau(\ell-1)}>0.
\end{align*}
Thus $\deg G_i^{(\tau)}>2g_i-2$, and by \cref{eq:ag-parameters} we have
\[
 \dim_{\F_{2^m}}\cC_i^{(\tau)}
 =\deg G_i^{(\tau)}+1-g_i.
\]
Substituting \cref{eq:tower-genus,eq:tower-degrees} and using \cref{r:bounds_bi}, we obtain
\[
 \frac{\dim_{\F_{2^m}}\cC_i^{(\tau)}}{n_i}
 =\frac1{\ell-1}\left(
 \frac\ell\tau-1+\frac{b_i^{(\tau)}+1}{\beta_i}\right)
 \geq\frac{\ell-(\tau-1)}{\tau(\ell-1)}.
\]
The designed distance from \cref{eq:ag-parameters} gives
\begin{align*}
 \frac{d(\cC_i^{(\tau)})}{n_i}
 \geq1-\frac{\deg G_i^{(\tau)}}{n_i}\geq\frac{(\tau-1)\ell-(\tau+1)}{\tau(\ell-1)}.
\end{align*}
Finally, from \cref{eq:ag-dual} and the computations above, we get
\begin{align*}
 \frac{d((\cC_i^{(\tau)})^\perp)}{n_i}
 &\geq \frac{\deg G_i^{(\tau)}-(2g_i-2)}{n_i}\geq \frac{\ell-(2\tau-1)}{\tau(\ell-1)},
\end{align*}
which proves \cref{eq:outer-dual-general}.
\end{proof}

\section{Alphabet reduction}\label{ss:map}
In this section, we introduce the maps we use to reduce the alphabet from $\F_{2^m}$ to $\F_2$, which, in concatenation terms, would be related to the \textit{inner codes}. This would be the analog of the multiplication-friendly embeddings considered in \cite{guruswami_good_transversal,nguyen_good_css_ccz}. For a motivating example, for $h=3$, suppose that an additive map
$\phi:\F_{2^m}\to\F_2^s$ satisfies
\[
 \sum_{p=1}^s\phi_p(a)\phi_p(b)\phi_p(c)=\Tr(a^2bc+ab^2c+abc^2)
 =\Tr\bigl(abc(a+b+c)\bigr).
\]
There are three binary factors, while the polynomial inside the trace is homogeneous of degree four. Thus, if $\cC\subset\F_{2^m}^n$ is a linear code satisfying $\cC^{\star3}\subset\cC^\perp$, then $\sum_{j=1}^{n}c_j^{(1)} c_j^{(2)}  c_j^{(3)}c_j^{(4)}=0$ for any $c^{(1)},c^{(2)},c^{(3)},c^{(4)}\in \cC$, and the block-wise image
\(
 \Phi(\cC)
 =\{(\phi(c_1),\ldots,\phi(c_n)):c\in\cC\}
\)
satisfies $\Phi(\cC)^{\star2}\subset\Phi(\cC)^\perp$.

\begin{remark}
Note that if we take $c=a$, we get 
\(
 \sum_{p=1}^s\phi_p(a)\phi_p(b)\phi_p(a)=\sum_{p=1}^s\phi_p(a)\phi_p(b)\), and 
\(
 \Tr\bigl(aba(a+b+a)\bigr)=\Tr\bigl((a^2b^2)\bigr)=\Tr(ab).
\)
Thus, both expressions reduce to products of $2$ symbols. Moreover, if we had taken $c=b$, we would have obtained the same result. If $a=b=c$, we get \(
 \sum_{p=1}^s\phi_p(a)\phi_p(a)\phi_p(a)=\sum_{p=1}^s\phi_p(a)\) and \(
 \Tr\bigl(aaa(a+a+a)\bigr)=\Tr(a).
\)
This behavior motivates the definition in \cref{eq:canonical-overlap-product}: the trace expressions are compatible with the reduction of repeated factors in binary overlaps. A natural alternative such as $\Tr(abc)$ does not, in general, have this compatibility. Indeed, binary overlaps necessarily satisfy
\(
 \sum_{p=1}^s\phi_p(a)^2\phi_p(b)
 =\sum_{p=1}^s\phi_p(a)\phi_p(b)^2,
\)
whereas $\Tr(a^2b)$ and $\Tr(ab^2)$ need not agree.
\end{remark}

The following result lets us transfer the required Schur-product condition to the binary image. We state it in a form that applies to all the constructions considered below.

\begin{lemma}\label{lem:schur-transfer}
Let $r\geq2$ and $\tau\geq2$, and let $\phi:\F_{2^m}\to\F_2^s$ be an additive map. Suppose that there is a homogeneous polynomial $P\in\F_{2^m}[X_1,\ldots,X_r]$ of degree $\tau$ such that
\begin{equation}\label{eq:general-map-identity}
 \sum_{p=1}^s\prod_{a=1}^r\phi_p(x_a)
 =\Tr\bigl(P(x_1,\ldots,x_r)\bigr)
\end{equation}
for all $x_1,\ldots,x_r\in\F_{2^m}$. Let $\cC\subset\F_{2^m}^n$ be a linear code satisfying $\cC^{\star(\tau-1)}\subset\cC^\perp$. Then
\(
 \Phi(\cC)
 =\{(\phi(c_1),\ldots,\phi(c_n)):c\in\cC\}
\)
satisfies
\[
 \Phi(\cC)^{\star(r-1)}\subset\Phi(\cC)^\perp.
\]
\end{lemma}

\begin{proof}
For $c^{(1)},\ldots,c^{(r)}\in\cC$, let $z=P(c^{(1)},\ldots,c^{(r)})$, evaluated coordinate-wise. Every monomial of $P$ has total degree $\tau$, so its evaluation is a scalar multiple of a Schur product of exactly $\tau$ codewords of $\cC$, allowing repetitions. Thus $z\in\cC^{\star\tau}$. The assumption $\cC^{\star(\tau-1)}\subset\cC^\perp$ implies that $\sum_{i=1}^n z_i=0$. Consequently,
\[
 \sum_{i=1}^n\sum_{p=1}^s
 \prod_{a=1}^r\phi_p(c_i^{(a)})
 =
 \Tr\!\left(\sum_{i=1}^n z_i\right)
 =0,
\]
which proves the result.
\end{proof}

\begin{remark}\label{rem:inhomogeneous-schur-transfer}
In \cref{lem:schur-transfer}, the condition that $P$ is homogeneous of degree $\tau$ may be replaced by $\deg P\leq\tau$ if one additionally assumes that $\one\in\cC$. Indeed, the evaluation of a monomial of total degree $d\leq\tau$ belongs to $\cC^{\star\tau}$ after inserting $\tau-d$ copies of $\one$ into the Schur product. In particular, the AG codes $\cC_i^{(\tau)}$ from \cref{prop:outer} contain $\one$.
\end{remark}

We now show that maps satisfying the hypotheses of \cref{lem:schur-transfer} exist for every number of factors. For $j\geq1$, define
\begin{equation}\label{eq:canonical-overlap-product}
 H_j(X_1,\ldots,X_j)
 =\prod_{\substack{S\subset\{1,\ldots,j\}\\ |S|\text{ odd}}}
 \left(\sum_{a\in S}X_a\right).
\end{equation}
This polynomial is symmetric, additive in each variable in characteristic two, and homogeneous of degree $2^{j-1}$. 
\begin{example}
For $j=2,3,4$ we have
\begin{align*}
 H_2(a,b)&=ab,\\
 H_3(a,b,c)&=abc(a+b+c),\\
 H_4(a,b,c,d)&=abcd(a+b+c)(a+b+d)(a+c+d)(b+c+d).
\end{align*}
\end{example}

These properties and the following proposition are proved directly in \cref{app:general-map}.

\begin{proposition}\label{prop:higher-map}
Let $j\geq2$ and $m\geq1$. There is an additive injective map
$\phi:\F_{2^m}\longrightarrow\F_2^s$, with
$s\leq\sum_{a=1}^{\min\{j,m\}}\binom ma$, such that
\begin{equation}\label{eq:higher-map-parity}
 \wt(\phi(x))=\Tr(x)\bmod2.
\end{equation}
Moreover, for every $x_1,\ldots,x_j\in\F_{2^m}$,
\begin{equation}\label{eq:higher-map-identity}
 \sum_{p=1}^s\prod_{a=1}^j\phi_p(x_a)
 =\Tr\bigl(H_j(x_1,\ldots,x_j)\bigr).
\end{equation}
In particular, the identity and homogeneity condition in \cref{lem:schur-transfer}
hold with $r=j$ and $\tau=2^{j-1}$.
\end{proposition}

The length bound in \cref{prop:higher-map} is not intended to be optimal. The explicit maps in \cref{sec:families} satisfy the required identities with much smaller lengths, and give better parameters for the $T$ gate.

\begin{remark}\label{rem:four-cases}
\Cref{prop:outer} is independent of the particular binary map used afterward. To compare the different cases, let $j$ denote the number of binary factors, so that the desired binary condition is $\mathcal D^{\star(j-1)}\subset\mathcal D^\perp$, and let $\tau$ be the corresponding parameter in \cref{lem:schur-transfer,prop:outer}. The values provided by \cref{prop:higher-map} for the two constructions below, together with their $h=3$ specializations, are
\[
\begin{array}{c|c|c}
 \text{construction}&j&\tau\\ \hline
 T\text{ with a correction}&3&4\\
 T\text{ without a correction}&4&8\\
 R_h\text{ with a correction}&h&2^{h-1}\\
 R_h\text{ without a correction}
 &2^{h-1}&2^{2^{h-1}-1}
\end{array}
\]
In each row, the outer condition is $\cC^{\star(\tau-1)}\subset\cC^\perp$. The first two rows are the cases $h=3$ used in \cref{sec:families}.
\end{remark}

\section{Construction with a lower-level correction}\label{s:concatenate}
In this section, we describe the puncturing and binary concatenation techniques we use to derive the binary codes for the CSS construction. Related constructions based on puncturing and shortening classical codes appear in \cite{guruswami_good_transversal,nguyen_good_css_ccz,wills_asymptotically_good_non_clifford}. We allow a correction from the preceding level of the Clifford hierarchy, as in \cref{lem:higher-correction}.

\subsection{Bounds for the distances}

We first separate the distance argument, which is the same for the two constructions in this paper. We start by recalling the notion of an adjoint map. For every additive map $\phi:\F_{2^m}\to\F_2^s$ and every $y\in\F_2^s$, the map $a\mapsto y\cdot\phi(a)$ is an $\F_2$-linear functional on $\F_{2^m}$. By the non-degeneracy of the trace pairing
\cite[Thm.~2.24]{lidl_niederreiter_finite_fields}, there is a unique
$\rho_\phi(y)\in\F_{2^m}$ such that
\begin{equation}\label{eq:adjoint}
 y\cdot\phi(a)=\Tr\bigl(\rho_\phi(y)a\bigr)
\end{equation}
for every $a\in\F_{2^m}$. Uniqueness also shows that $\rho_\phi:\F_2^s\to\F_{2^m}$ is $\F_2$-linear. Indeed, we have $(y+y')\cdot \phi(a)=\operatorname{Tr}\bigl(\rho_\phi(y)a\bigr)+\operatorname{Tr}\bigl(\rho_\phi(y')a\bigr)=\operatorname{Tr}\bigl(\left(\rho_\phi(y)+\rho_\phi(y')\right)a\bigr)$, and thus $\rho_\phi(y)+\rho_\phi(y')=\rho_\phi(y+y')$ by uniqueness. The map $\rho_\phi$ will be key in the next lemma to bound the $Z$-distance of the quantum codes.

\begin{lemma}\label{lem:puncturing-distances}
Let $\F_{2^m}/\F_2$ be a finite extension, let $\cC\subset\F_{2^m}^n$ be an $\F_{2^m}$-linear code, and let $\phi:\F_{2^m}\to\F_2^s$ be additive and injective. Let
\[
 d_{\rm in}=\min_{a\neq0}\wt(\phi(a)).
\]
Choose $P$ coordinates in an information set of $\cC$, consider $L=n-P$, and let $\cA_1,\cA_2\subset\F_{2^m}^L$ be respectively the puncture and the shortening of $\cC$ on the selected coordinates. Let
\[
 \Phi:\F_{2^m}^L\longrightarrow\F_2^{sL}, \; \Phi(a_1,\dots,a_L)=(\phi(a_1),\dots,\phi(a_L)). 
\]
Suppose that $C_2\subset C_1\subset\F_2^{sL}$ satisfy
\[
 C_2=\Phi(\cA_2),
 \;
 C_1\subset\Phi(\cA_1).
\]
Then
\begin{align}
 d_X\bigl(\operatorname{CSS}(C_1,C_2)\bigr)
 &\geq d_{\rm in}\bigl(d(\cC)-P\bigr),
 \label{eq:puncturing-distance-x}\\
 d_Z\bigl(\operatorname{CSS}(C_1,C_2)\bigr)
 &\geq d(\cC^\perp)-P.
 \label{eq:puncturing-distance-z}
\end{align}
\end{lemma}

\begin{proof}
Every word of $C_1\setminus C_2$ is $\Phi(a)$ for some $a\in\cA_1\setminus\cA_2$. Puncturing $P$ coordinates decreases distance by at most $P$, and every nonzero block $\phi(a_j)$ has weight at least $d_{\rm in}$. This proves \cref{eq:puncturing-distance-x}.

For the bound for $d_Z$, define the block-wise adjoint map
\[
 \rho_{\phi,L}:\F_2^{sL}\longrightarrow\F_{2^m}^L,
 \;
 \rho_{\phi,L}(y)
 =
 \bigl(\rho_\phi(y^{(1)}),\ldots,\rho_\phi(y^{(L)})\bigr),
\]
where $y=(y^{(1)},\ldots,y^{(L)})$ and $y^{(r)}\in\F_2^s$. Consider
\[
 y\in C_2^\perp\setminus C_1^\perp.
\]
We first show that $\rho_{\phi,L}(y)\in\cA_2^\perp$. Since $y\in C_2^\perp$ and $C_2=\Phi(\cA_2)$, for every $a=(a_1,\ldots,a_L)\in\cA_2$ we have
\begin{equation}\label{eq:adjoint-prod-general}
 0
 =y\cdot\Phi(a)
 =\sum_{r=1}^L y^{(r)}\cdot\phi(a_r)
 =\Tr\left(\sum_{r=1}^L\rho_\phi(y^{(r)})a_r\right).
\end{equation}
The code $\cA_2$ is $\F_{2^m}$-linear. Hence, for every
$\lambda\in\F_{2^m}$, we may replace $a$ by $\lambda a$ and obtain
\[
 0=\Tr\left(\lambda\sum_{r=1}^L\rho_\phi(y^{(r)})a_r\right).
\]
By the non-degeneracy of the trace, this implies
\[
 \sum_{r=1}^L\rho_\phi(y^{(r)})a_r=0
\]
for every $a\in\cA_2$. Therefore, $\rho_{\phi,L}(y)\in\cA_2^\perp$.

We next show that $\rho_{\phi,L}(y)\notin\cA_1^\perp$. Suppose, to the contrary, that $\rho_{\phi,L}(y)\in\cA_1^\perp$. Then, as in \cref{eq:adjoint-prod-general}, we would get that $y$ is orthogonal to $\Phi(\cA_1)$. Since $C_1\subset\Phi(\cA_1)$, this would imply $y\in C_1^\perp$, a contradiction with the choice of $y$.

Consequently, $\rho_{\phi,L}(y)\in\cA_2^\perp\setminus\cA_1^\perp$. If the $r$th coordinate of $\rho_{\phi,L}(y)$ is nonzero, then
$y^{(r)}\neq0$. Hence
\[
 \wt(y)
 \geq\bigl|\{r:y^{(r)}\neq0\}\bigr|
 \geq\wt\bigl(\rho_{\phi,L}(y)\bigr).
\]
The dual of a shortening is a puncturing of the dual, and puncturing at $P$ coordinates can decrease distance by at most $P$. Thus $d(\cA_2^\perp)\geq d(\cC^\perp)-P$.
It follows that
\[
 \wt(y)
 \geq\wt\bigl(\rho_{\phi,L}(y)\bigr)
 \geq d(\cA_2^\perp)
 \geq d(\cC^\perp)-P,
\]
which proves \cref{eq:puncturing-distance-z}.
\end{proof}

\subsection{The puncturing construction}
In this section, we give the puncturing construction used to obtain binary quantum codes. This is the construction that, instantiated with the codes from \cref{sec:outer}, provides asymptotically good quantum codes. We start with a technical definition. 
\begin{definition}\label{def:frame}
Let $h\geq3$ and let $\phi:\F_{2^m}\to\F_2^s$ be additive and injective.
A \emph{compatible set} for $\phi$ and $h$ is a tuple
$u_1,\ldots,u_t\in\F_{2^m}$ such that
\begin{equation*}\label{eq:general-frame-pair}
 \phi(u_a)\cdot\phi(u_b)=\delta_{ab}
\end{equation*}
for any $1\leq a,b\leq t$, and, for $3\leq r \leq h$,
\begin{equation}\label{eq:general-frame-overlap}
 \wt\bigl(\phi(u_{a_1})\star\cdots\star\phi(u_{a_r})\bigr)=0\bmod2
\end{equation}
whenever $a_1,\ldots,a_r$ are pairwise distinct.
\end{definition}

The first condition implies that a compatible set is $\F_2$-linearly independent, so $t\leq m$. It also implies that $\wt(\phi(u_a))$ is odd, which will make the corresponding rows odd-weighted. For $h=3$, the second condition only concerns triples of distinct elements.

\begin{remark}
There is always a compatible set of size $t=1$ for the maps $\phi$ from \cref{prop:higher-map}. Indeed, the condition from \cref{eq:general-frame-overlap} is vacuous, and we only require $\phi(u_1)\cdot\phi(u_1)=1$. By \cref{eq:higher-map-parity}, this holds whenever $\Tr(u_1)=1$. Since the absolute trace is a nonzero binary linear functional, exactly half of the field elements have trace one, and any of them can be chosen as $u_1$. Nevertheless, finding compatible sets with larger cardinality yields quantum codes with better parameters.
\end{remark}

\begin{theorem}\label{thm:general-puncture}
Let $h\geq3$ and $\tau\geq2$, let $\cC\subset\F_{2^m}^n$ be an $\F_{2^m}$-linear code satisfying $\cC^{\star(\tau-1)}\subset\cC^\perp$, and let $\phi:\F_{2^m}\to\F_2^s$ be additive and injective, satisfying the identity and homogeneity condition in \cref{lem:schur-transfer} with $r=h$. Write
\[
 d_{\rm in}=\min_{a\neq0}\wt(\phi(a)),
\]
and let $u_1,\ldots,u_t$ be a compatible set for $\phi$ and $h$.

Choose $P<\min\{\dim_{\F_{2^m}}\cC,d(\cC)\}$ coordinates in an information set. Then there are binary codes $C_2\subset C_1\subset\F_2^{s(n-P)}$ such that there is an $h$-orthogonal generator matrix $G$ for $C_1$ with exactly $tP$ odd-weight rows, and the remaining rows span $C_2$. Moreover,
\begin{align}
 \dim_{\F_2}(C_1/C_2)&=tP,
 \notag\\
 d_X\bigl(\operatorname{CSS}(C_1,C_2)\bigr)
 &\geq d_{\rm in}\bigl(d(\cC)-P\bigr),
 \label{eq:general-x}\\
 d_Z\bigl(\operatorname{CSS}(C_1,C_2)\bigr)
 &\geq d(\cC^\perp)-P.
 \label{eq:general-z}
\end{align}
A diagonal correction $U\in\mathfrak C_{h-1}$, computable from $G$, satisfies
\begin{equation*}\label{eq:logical-action}
 UR_h^{\otimes s(n-P)}\big|_{\operatorname{CSS}(C_1,C_2)}
 =\overline R_h^{\otimes tP}.
\end{equation*}
\end{theorem}

\begin{proof}
After permuting coordinates, write a generator matrix of $\cC$ as
\[
 \begin{pmatrix}
 I_P&H_1\\
 0&H_0
 \end{pmatrix}.
\]
Let $L=n-P$ and define
\[
 \cA_1:=\operatorname{rowsp}_{\F_{2^m}}(H_1,H_0),
 \;
 \cA_2:=\operatorname{rowsp}_{\F_{2^m}}(H_0).
\]
Thus $\cA_1$ is the puncture and $\cA_2$ the shortening of $\cC$ on the selected coordinates. Let
\[
 \Phi:\F_{2^m}^L\longrightarrow\F_2^{sL},
 \;
 \Phi(a_1,\ldots,a_L)=(\phi(a_1),\ldots,\phi(a_L)).
\]
Set $C_2=\Phi(\cA_2)$. If $h_p$ is the $p$th row of $H_1$, define
\begin{equation*}\label{eq:general-logical-rows}
 g_{p,a}=\Phi(u_ah_p),
 \; 1\leq p\leq P,\; 1\leq a\leq t,
\end{equation*}
and
\begin{equation*}\label{eq:general-C1}
 C_1=C_2+\Span_{\F_2}\{g_{p,a}:1\leq p\leq P, 1\leq a\leq t\}.
\end{equation*}

Let $G$ be the matrix whose first rows are given by $g_{p,a}$, and the rest of its rows are given by a basis of $C_2$. We will prove the overlap conditions for this matrix. First, note that a word $v\in\cA_2$ has the shortened lift
\[
 \widehat v=(0,\ldots,0\mid v)\in\cC,
\]
whereas the lift associated with $g_{p,a}$ is
\[
 \lambda_{p,a}=u_a(e_p\mid h_p)\in\cC.
\]
Let $\pi:\cC\to\F_{2^m}^L$ puncture the selected coordinates.
The lift is unique since both $\pi|_{\cC}$ and $\Phi$ are injective.
Let
\[
 \widehat{\Phi}:\F_{2^m}^{n}\longrightarrow\F_2^{sn},
 \;
 \widehat{\Phi}(a_1,\ldots,a_n)
 =(\phi(a_1),\ldots,\phi(a_n)).
\]
By \cref{lem:schur-transfer}, the image $\mathcal D=\widehat{\Phi}(\cC)$ satisfies $\mathcal D^{\star(h-1)}\subset\mathcal D^\perp$. Consequently, the product of any $r\leq h$ words of $\mathcal D$ has even weight (in the $h$-fold product, we may repeat one of the words when $r<h$). For clarity, we start with two codewords, and then move to the general case. For $x,y\in\cC$, we have
\begin{equation}\label{eq:punctured-pair}
 \Phi(\pi(x))\cdot\Phi(\pi(y))
 =\sum_{b=P+1}^n\phi(x_b)\cdot\phi(y_b)
 =\sum_{b=1}^P\phi(x_b)\cdot\phi(y_b).
\end{equation}
This is zero if either $\pi(x)$ or $\pi(y)$ belongs to the shortening $\cA_2$, while if $\Phi(\pi(x))=g_{p_1,a}$ and $\Phi(\pi(y))=g_{p_2,b}$ (i.e., $x=\lambda_{p_1,a}$ and $y=\lambda_{p_2,b}$), we have
\begin{equation}\label{eq:product_gs}
 g_{p_1,a}\cdot g_{p_2,b}
 =\phi(u_a)\cdot\phi(u_b)\delta_{p_1p_2}
 =\delta_{ab}\delta_{p_1p_2}.
\end{equation}
This proves the overlap condition for pairs of rows, and the fact that the first $tP$ rows of $G$ have odd weight. The remaining rows have even weight by applying \cref{eq:punctured-pair} to a shortened lift twice.

Similarly, for the overlap of $r$ distinct rows, with $3\leq r\leq h$, take their lifts $x^{(1)},\ldots,x^{(r)}\in\cC$. Since their binary images belong to $\mathcal D$, we have
\begin{equation}\label{eq:punctured-overlap}
 \sum_{b=P+1}^n\sum_{\ell=1}^s
 \prod_{a=1}^r\phi_\ell(x_b^{(a)})
 =
 \sum_{b=1}^P\sum_{\ell=1}^s
 \prod_{a=1}^r\phi_\ell(x_b^{(a)}).
\end{equation}
If one of $\pi(x^{(1)}),\ldots,\pi(x^{(r)})$ belongs to $\cA_2$, this sum is equal to 0. It remains to consider the case in which
\[
 \Phi(\pi(x^{(b)}))=g_{p_b,a_b},
 \; 1\leq b\leq r.
\]
Then $x^{(b)}=\lambda_{p_b,a_b}$, and its entries in the removed
coordinates are
\[
 x_c^{(b)}=u_{a_b}\delta_{cp_b},
 \; 1\leq c\leq P.
\]
Consequently, the right-hand side of \cref{eq:punctured-overlap} is
\[
 \left(\sum_{\ell=1}^s
 \prod_{b=1}^r\phi_\ell(u_{a_b})\right)\delta_{p_1\cdots p_r},
\]
where $\delta_{p_1\cdots p_r}=1$ if $p_1=\cdots=p_r$, and is zero otherwise. Thus, if the indices are not equal, the overlap is zero. If $p_1=\cdots=p_r$, then the indices $a_1,\ldots,a_r$ are pairwise distinct, and \cref{eq:general-frame-overlap} gives
\[
 \sum_{\ell=1}^s\prod_{b=1}^r\phi_\ell(u_{a_b})=0.
\]
This proves the required overlap conditions.

By \cref{eq:product_gs}, we have $g_{p_1,a}\cdot g_{p_2,b}=\delta_{ab}\delta_{p_1p_2}$. Thus, the first $tP$ rows are linearly independent. If we had $\sum_{p,a}\mu_{p,a}g_{p,a}\in C_2$, multiplying by $g_{c,b}$ we get $\mu_{c,b}=0$, for all $c,b$ (recall we have already proven the inner product of any two different rows is 0, in particular, $g_{c,b}$ is orthogonal to $C_2$). Thus $G$ is full rank.

Since $C_2=\Phi(\cA_2)$ and $C_1\subset\Phi(\cA_1)$, \cref{lem:puncturing-distances} gives \cref{eq:general-x,eq:general-z}. The logical action follows from \cref{lem:higher-correction}.
\end{proof}

\subsection{Asymptotic parameters}

Combining \cref{prop:outer,thm:general-puncture} gives the following
general asymptotic construction.

\begin{theorem}\label{thm:general-asymptotic}
Let $h\geq3$, let $m\geq6$ be even, and put $\ell=2^{m/2}$. Let $\tau$ be a power of $2$ such that $4\leq\tau\leq\ell/2$, and let $\{\cC_i^{(\tau)}\}_i$ be the family of codes over $\F_{2^m}$ from \cref{prop:outer}, with lengths $n_i$. Let $\phi:\F_{2^m}\to\F_2^s$ be additive and injective, satisfying the identity and homogeneity condition in \cref{lem:schur-transfer} with $r=h$, write
\[
 d_{\rm in}=\min_{a\neq0}\wt(\phi(a)),
\]
and let $u_1,\ldots,u_t$ be a compatible set for $\phi$ and $h$.
Fix $0<\kappa<B_\ell^{(\tau)}$ and let $P_i=\lfloor\kappa n_i\rfloor$. There is an explicit family $\{\cQ_i\}_i$ with parameters $[[N_i,K_i,\Delta_i]]_2$, $N_i=s(n_i-P_i)$, $K_i=tP_i$, for which
\begin{align*}
 d_X(\cQ_i)
 &\geq d_{\rm in}\bigl(A_\ell^{(\tau)} n_i-P_i\bigr),\\
 d_Z(\cQ_i)
 &\geq B_\ell^{(\tau)} n_i-P_i,
\end{align*}
and
\begin{align*}
 \liminf_{i\to \infty}\frac{K_i}{N_i}
 &\geq\frac{t\kappa}{s(1-\kappa)},\\
 \liminf_{i\to \infty}\frac{\Delta_i}{N_i}
 &\geq\frac{B_\ell^{(\tau)}-\kappa}{s(1-\kappa)}.
\end{align*}
Moreover, these codes have $h$-orthogonal generator matrices, and, for each $i$, a diagonal correction $U_i\in\mathfrak C_{h-1}$ can be computed from the defining $h$-orthogonal matrix such that $U_iR_h^{\otimes N_i}$ acts as logical $\overline R_h^{\otimes K_i}$.
\end{theorem}
\begin{proof}
By \cref{prop:outer} and the assumption
$\kappa<B_\ell^{(\tau)}$, we have
\[
 P_i<\min\{\dim_{\F_{2^m}}\cC_i^{(\tau)},d(\cC_i^{(\tau)})\},
\]
since
\[
 B_\ell^{(\tau)}<
 \rho_\ell^{(\tau)}
 \text{ and }
 B_\ell^{(\tau)}<A_\ell^{(\tau)}.
\]
We may therefore apply \cref{thm:general-puncture} to $\cC_i^{(\tau)}$ with
$P=P_i$. This gives
\[
 N_i=s(n_i-P_i),\; K_i=tP_i,
\]
and
\begin{align*}
 d_X(\cQ_i)
 &\geq d_{\rm in}\bigl(d(\cC_i^{(\tau)})-P_i\bigr)
 \geq d_{\rm in}\bigl(A_\ell^{(\tau)} n_i-P_i\bigr),\\
 d_Z(\cQ_i)
 &\geq d((\cC_i^{(\tau)})^\perp)-P_i
 \geq B_\ell^{(\tau)} n_i-P_i.
\end{align*}
Moreover,
\[
 d_{\rm in}(A_\ell^{(\tau)} n_i-P_i)
 \geq A_\ell^{(\tau)} n_i-P_i
 >B_\ell^{(\tau)} n_i-P_i,
\]
so the smaller of the two displayed distance bounds is the one for
$d_Z(\cQ_i)$. Hence $\Delta_i\geq B_\ell^{(\tau)} n_i-P_i$. Dividing by $N_i$ and using $P_i/n_i\to\kappa$ gives the stated
asymptotic bounds. The statement about the logical action follows directly
from \cref{thm:general-puncture}.
\end{proof}

\begin{corollary}\label{thm:higher-with-correction}
For every fixed $h\geq3$, there is an explicit asymptotically good family of binary CSS codes with parameters $[[N_i,K_i,\Delta_i]]_2$, with an $h$-orthogonal generator matrix, and for which the gate $R_h^{\otimes N_i}$, followed by a diagonal correction from $\mathfrak C_{h-1}$, acts logically as $\overline R_h^{\otimes K_i}$.
\end{corollary}

\begin{proof}
Take $\tau=2^{h-1}$ and an even $m\geq2h$, so that
$\ell=2^{m/2}\geq2\tau$. Apply \cref{prop:higher-map} with $j=h$.
By \cref{eq:higher-map-parity}, there is an element $u$ with odd image
weight, which gives a compatible set of size one. Thus
\cref{thm:general-asymptotic} applies with $t=1$ and any
$0<\kappa<B_\ell^{(\tau)}$.
\end{proof}

\section{A correction-free construction}\label{sec:correction-free}
The families in \cref{s:concatenate} implement the desired logical gate after a diagonal correction from the preceding level. In this section, we strengthen the condition on the binary image so that no correction is required.

\subsection{Weighted phases}

The following weighted-phase criterion is a variant of \cite[Lem.~III.2 and~V.3]{haah_towers_divisible_codes}, allowing arbitrary integer physical weights and logical gates $R_h$ or $R_h^{-1}$. We include a proof for completeness.

\begin{lemma}\label{lem:higher-weighted-phase}
Let $h\geq3$ and let $g^{(1)},\ldots,g^{(M)}\in\F_2^N$ be independent, with $K\leq M$, and set
\[
 C_1=\Span\{g^{(1)},\ldots,g^{(M)}\},\;
 C_2=\Span\{g^{(K+1)},\ldots,g^{(M)}\}.
\]
Suppose that, for some $\nu\in \mathbb{Z}^N$, we have that
\[
 \wt_\nu(g^{(a)})=\epsilon_a\bmod2^h,
 \; \epsilon_a\in\{1,2^h-1\},
 \; a\leq K,
\]
that $\wt_\nu(g^{(a)})=0\bmod2^h$ for $a>K$, and that the weighted overlap of every $r$ distinct rows is zero modulo $2^{h-r+1}$ for $2\leq r\leq h$. Then
\[
 \bigotimes_{p=1}^N R_h^{\nu_p}
\]
preserves the corresponding CSS code and acts on its logical basis as
\[
 \overline R_h^{\epsilon_1}\otimes\cdots\otimes
 \overline R_h^{\epsilon_K}.
\]
\end{lemma}

\begin{proof}
For $u=(u_1,\ldots,u_M)\in\F_2^M$, apply \cref{lem:weight-expansion} modulo $2^h$. The term corresponding to a set of $r$ distinct rows has coefficient $(-2)^{r-1}$, and hence vanishes by the assumed divisibility when $2\leq r\leq h$; terms with $r>h$ vanish from the coefficient alone. Therefore
\[
 \wt_\nu\left(\sum_{a=1}^M u_ag^{(a)}\right)
 \equiv\sum_{a=1}^K\epsilon_au_a\bmod2^h.
\]
For $\bar u=(u_1,\ldots,u_K)\in\F_2^K$, the corresponding basis state is
\[
 \lvert\bar u\rangle
 =
 \frac{1}{\sqrt{|C_2|}}
 \sum_{y\in C_2}
 \left|y+\sum_{a=1}^K u_ag^{(a)}\right\rangle.
\]
Since
\[
 \left(\bigotimes_{p=1}^N R_h^{\nu_p}\right)\lvert x\rangle
 =\zeta_{2^h}^{\wt_\nu(x)}\lvert x\rangle,
\]
the preceding congruence shows that
\[
 \left(\bigotimes_{p=1}^N R_h^{\nu_p}\right)\lvert\bar u\rangle
 =
 \zeta_{2^h}^{\sum_{a=1}^K\epsilon_au_a}
 \lvert\bar u\rangle,
\]
which is precisely the action of
\(
 \overline R_h^{\epsilon_1}\otimes\cdots\otimes
 \overline R_h^{\epsilon_K}.
\)
In particular, $\bigotimes_{p=1}^N R_h^{\nu_p}$ preserves
$\operatorname{CSS}(C_1,C_2)$.
\end{proof}

We obtain the weighted congruences before puncturing from the following result. Its proof is given in \cref{app:higher-completion}.

\begin{lemma}\label{lem:higher-odd-completion}
Let $h\geq2$ and let $D\subset\F_2^N$ be a linear code such that
\[
 D^{\star(2^{h-1}-1)}\subset D^\perp.
\]
Then one can compute
\[
 \nu\in\{1,3,\ldots,2^h-1\}^N
\]
such that, for $1\leq r\leq h$,
\begin{equation}\label{eq:higher-weighted-congruences}
 \wt_\nu(x^{(1)}\star\cdots\star x^{(r)})
 =0\bmod 2^{h-r+1}
\end{equation}
for all $x^{(1)},\ldots,x^{(r)}\in D$.
\end{lemma}

\begin{example}
For $h=3$, the hypothesis in \cref{lem:higher-odd-completion} is $D^{\star3}\subset D^\perp$, and the moduli in \cref{eq:higher-weighted-congruences} are $8$, $4$, and $2$ for one, two, and three factors, respectively. These are the divisibility conditions used in \cref{lem:higher-weighted-phase}. We will use this to construct codes with a weighted transversal $T$ gate with no subsequent Clifford correction (also see \cite[Thm.~4 and Cor.~4]{rengaswamyOptimalityCSST}).
\end{example}

\subsection{The correction-free construction}
We now combine the preceding results to obtain the main result of this section.
\begin{theorem}\label{thm:correction-free-general}
Let $h\geq3$, let $m\geq6$ be even, and let $\ell=2^{m/2}$. Let $\tau$ be a power of $2$ such that $4\leq\tau\leq\ell/2$, and let $\{\cC_i^{(\tau)}\}_i$ be the family from \cref{prop:outer}. Let $\phi:\F_{2^m}\to\F_2^s$ be additive and injective, satisfying the identity and homogeneity condition in \cref{lem:schur-transfer} with $r=2^{h-1}$. Suppose that $\wt(\phi(u))$ is odd for some $u\in\F_{2^m}$, and set
\[
d_{\rm in}=\min_{a\neq0}\wt(\phi(a)).
\]
For every $0<\kappa<B_\ell^{(\tau)}/(2^{h-1}-1)$ there is an explicit family $\{\cQ_i\}_i$ of CSS codes with parameters $[[N_i,K_i,\Delta_i]]_2$, $K_i=\lfloor\kappa n_i\rfloor$, with $2^{h-1}$-orthogonal generator matrices, vectors $\nu_i'\in\{1,3,\ldots,2^h-1\}^{N_i}$, and integers $P_i$ with $K_i\leq P_i\leq(2^{h-1}-1)K_i$ such that
\begin{align}
N_i&=s(n_i-P_i),
\label{eq:correction-free-length}\\
d_X(\cQ_i)
&\geq d_{\rm in}(A_\ell^{(\tau)}n_i-P_i),
\label{eq:correction-free-x}\\
d_Z(\cQ_i)
&\geq B_\ell^{(\tau)} n_i-P_i.
\label{eq:correction-free-z}
\end{align}
In a suitable choice of logical basis,
$\bigotimes_{p=1}^{N_i}R_h^{(\nu_i')_p}$ acts as
$\overline R_h^{\otimes K_i}$ up to a global phase, with no subsequent
correction.
Moreover,
\begin{align}
\liminf_{i\to \infty}\frac{K_i}{N_i}
&\geq\frac{\kappa}{s(1-\kappa)},
\label{eq:correction-free-rate}\\
\liminf_{i\to \infty}\frac{\Delta_i}{N_i}
&\geq\frac{B_\ell^{(\tau)}-(2^{h-1}-1)\kappa}
{s\bigl(1-(2^{h-1}-1)\kappa\bigr)}.
\label{eq:correction-free-relative-distance}
\end{align}
\end{theorem}
\begin{proof}
Let
\[
\widehat{\Phi}_i:\F_{2^m}^{n_i}\longrightarrow\F_2^{sn_i},
\;
 \widehat{\Phi}_i(a_1,\ldots,a_{n_i})=(\phi(a_1),\ldots,\phi(a_{n_i})),
\]
and consider $\mathcal D_i=\widehat{\Phi}_i(\cC_i^{(\tau)})$. By \cref{lem:schur-transfer}, applied with $r=2^{h-1}$, we have $\mathcal D_i^{\star(2^{h-1}-1)}\subset\mathcal D_i^\perp$, and \cref{lem:higher-odd-completion} gives $\nu_i\in\{1,3,\ldots,2^h-1\}^{sn_i}$. Choose $u\in\F_{2^m}$ for which $\wt(\phi(u))$ is odd, as in the hypothesis. For each coordinate $1\leq p\leq n_i$, let
\[
\omega_p=\wt_{\nu_i^{(p)}}(\phi(u))\bmod2^h,
\]
where $\nu_i^{(p)}$ is the restriction to the $p$th inner block of length $s$. Since the entries of $\nu_i$ are odd and $\wt(\phi(u))$ is odd too, we have $\omega_p\in\{1,3,\ldots,2^h-1\}$.

Let $I_i$ be an information set of $\cC_i^{(\tau)}$, and for
$a\in\{1,3,\ldots,2^h-1\}$ set
\[
I_{i,a}=\{p\in I_i:\omega_p\equiv a\bmod2^h\}.
\]
For each such $a$, choose $b_a\in\{1,\ldots,2^{h-1}-1\}$ such that
\[
b_aa\equiv1\ \text{or}\ -1\bmod2^h.
\]
Such a $b_a$ exists: take the smaller of the inverse of $a$ modulo $2^h$ and its negative. Divide each $I_{i,a}$ into disjoint sets of $b_a$ elements, discarding the final incomplete set. Any of the sets $J$ constructed in this way satisfies
\begin{equation*}\label{eq:group-residue}
\sum_{p\in J}\omega_p\equiv1\ \text{or}\ -1\bmod2^h.
\end{equation*}
\begin{example}
For $h=3$, this gives $b_1=b_7=1$ and $b_3=b_5=3$: we use the coordinates of residue $1$ or $7$ individually, and group those of residue $3$ or $5$ in triples, since $3+3+3\equiv1\bmod8$ and $5+5+5\equiv-1\bmod8$.
\end{example}
The total number of sets $J$ is
\[
\sum_{a\in\{1,3,\ldots,2^h-1\}}
\left\lfloor\frac{|I_{i,a}|}{b_a}\right\rfloor.
\]
We discard at most $b_a-1$ coordinates from each $I_{i,a}$, hence at most
\[
c_h=\sum_{a\in\{1,3,\ldots,2^h-1\}}(b_a-1)
\]
altogether. If $k_i=\dim_{\F_{2^m}}\cC_i^{(\tau)}$, we are considering at least $k_i-c_h$ coordinates, and each set has at most $2^{h-1}-1$ elements. Thus, the total number of sets is at least
\[
\left\lceil\frac{k_i-c_h}{2^{h-1}-1}\right\rceil
\geq\frac{\rho_\ell^{(\tau)}n_i-c_h}{2^{h-1}-1}.
\]
Moreover, $\rho_\ell^{(\tau)}-B_\ell^{(\tau)}=\frac1{\ell-1}>0$, and since
$\kappa<B_\ell^{(\tau)}/(2^{h-1}-1)$, we have
$\rho_\ell^{(\tau)}-(2^{h-1}-1)\kappa>0$. Hence, for every $i$ such that
\[
n_i\geq\frac{c_h}{\rho_\ell^{(\tau)}-(2^{h-1}-1)\kappa},
\]
if we rearrange the terms in that inequality, we get that the number of groups is greater than or equal to $\kappa n_i\geq\lfloor\kappa n_i\rfloor=K_i$. Thus, there are at least $K_i$ such groups for all sufficiently large $i$. Discarding finitely many initial terms and re-indexing the family, we may assume that this holds for every $i$. Retain $K_i$ of them and let $P_i$ be the number of coordinates they contain. Then $K_i\leq P_i\leq(2^{h-1}-1)K_i$, since each group has cardinality at most $2^{h-1}-1$. Also $P_i<\min\{k_i,d(\cC_i^{(\tau)}),d((\cC_i^{(\tau)})^\perp)\}$, because $P_i\leq (2^{h-1}-1)\kappa n_i<B_\ell^{(\tau)}n_i<\min\{\rho_\ell^{(\tau)},A_\ell^{(\tau)}\}n_i$. Let $S_i=\bigsqcup_{a=1}^{K_i}J_a$, with $\size{S_i}=P_i$, and, after permuting coordinates, write a generator matrix of $\cC_i^{(\tau)}$ in systematic form on $S_i$:
\[
\begin{pmatrix}
I_{P_i}&H_{1,i}\\
0&H_{0,i}
\end{pmatrix}.
\]
Put $L_i=n_i-P_i$ and
\[
\cA_{1,i}
=\operatorname{rowsp}_{\F_{2^m}}(H_{1,i},H_{0,i}),
\;
\cA_{2,i}
=\operatorname{rowsp}_{\F_{2^m}}(H_{0,i}),
\]
so that $\cA_{1,i}$ and $\cA_{2,i}$ are respectively the puncture and the shortening of $\cC_i^{(\tau)}$ on $S_i$. Define
\[
 \Phi_i:\F_{2^m}^{L_i}\longrightarrow\F_2^{sL_i}, \; \Phi_i(a_1,\dots,a_{L_i})=(\phi(a_1),\dots,\phi(a_{L_i})),
\]
and consider $C_{2,i}=\Phi_i(\cA_{2,i})$. For $p\in S_i$, let $h_p$ be the row of $H_{1,i}$ corresponding to $p$. For each selected set $J_a$, define
\[
\lambda_a
=u\sum_{p\in J_a}(e_p\mid h_p)
=\left(u\one_{J_a}\,\middle|\,
u\sum_{p\in J_a}h_p\right)\in\cC_i^{(\tau)}
\]
and
\[
g^{(a)}
=\Phi_i\left(u\sum_{p\in J_a}h_p\right)
\in\F_2^{sL_i}.
\]
Thus $g^{(a)}$ is obtained from
$\widehat g^{(a)}:=\widehat{\Phi}_i(\lambda_a)\in\mathcal D_i$ by deleting the
blocks indexed by $S_i$. Set
\[
C_{1,i}
=C_{2,i}+\Span_{\F_2}
\{g^{(1)},\ldots,g^{(K_i)}\}.
\]
Let $\nu_i'$ be the restriction of $\nu_i$ to the remaining $sL_i$
binary coordinates, and let
\[
\eta_a
\equiv\sum_{p\in J_a}\omega_p\bmod2^h,
\; \eta_a\in\{1,2^h-1\}.
\]
Since $\widehat g^{(a)}\in\mathcal D_i$, we have $\wt_{\nu_i}(\widehat g^{(a)})\equiv0\bmod2^h$ (recall that we choose $\nu_i$ from \cref{lem:higher-odd-completion}). Its deleted blocks are
equal to $\phi(u)$ precisely at the coordinates in $J_a$. Therefore
\begin{equation}\label{eq:logical-weight-after-puncture}
\wt_{\nu_i'}(g^{(a)})
\equiv-\sum_{p\in J_a}\omega_p
\equiv-\eta_a
=:\epsilon_a\bmod2^h,
\; \epsilon_a\in\{1,2^h-1\}.
\end{equation}
On the other hand, every $v\in C_{2,i}$ has a lift
$\widehat v\in\mathcal D_i$ that vanishes on all the deleted blocks.
Consequently,
\begin{equation}\label{eq:stabilizer-weight-after-puncture}
\wt_{\nu_i'}(v)
=\wt_{\nu_i}(\widehat v)
\equiv0\bmod2^h.
\end{equation}

The rest of the proof is similar to that of \cref{thm:general-puncture}. Let $G_i$ be the matrix whose first rows are given by $g^{(1)},\ldots,g^{(K_i)}$, and the rest of its rows are given by a basis of $C_{2,i}$. We will prove the overlap conditions for this matrix. By a binary lift of a row $x$, we mean the unique word $\widehat x\in\mathcal D_i$ that gives $x$ after deleting the blocks indexed by $S_i$.

First, note that a word $v\in\cA_{2,i}$ has the shortened lift
\(
 (0,\ldots,0\mid v)\in\cC_i^{(\tau)},
\)
so that the binary lift of $x=\Phi_i(v)\in C_{2,i}$ is
\(
 \widehat x
 =\widehat{\Phi}_i\bigl((0,\ldots,0\mid v)\bigr)
 =(\mathbf{0}_{sP_i}\mid x)\in\mathcal D_i,
\)
whereas the binary lift associated with $g^{(a)}$ is
\(
 \widehat g^{(a)}=\widehat{\Phi}_i(\lambda_a)\in\mathcal D_i.
\)
Let $\pi_i:\cC_i^{(\tau)}\to\F_{2^m}^{L_i}$ puncture the selected
coordinates. The lift is unique since both $\pi_i$ and $\Phi_i$ are
injective, the former because $P_i<d(\cC_i^{(\tau)})$.

For clarity, we start with two rows, and then move to the general case.
Let $x$ and $x'$ be two distinct rows of $G_i$, with binary lifts
$\widehat x,\widehat x'\in\mathcal D_i$. If either $x$ or $x'$ belongs
to $C_{2,i}$, its binary lift is zero on all the deleted blocks, so
$\widehat x\star\widehat x'$ is also zero there. It remains to consider
the case in which
\(
 x=g^{(a)}, x'=g^{(b)}, a\neq b.
\)
Then $\widehat x=\widehat{\Phi}_i(\lambda_a)$ and
$\widehat x'=\widehat{\Phi}_i(\lambda_b)$. On the deleted blocks,
$\widehat x$ is equal to $\phi(u)$ at the coordinates in $J_a$ and
zero elsewhere, while $\widehat x'$ is equal to $\phi(u)$ at the
coordinates in $J_b$ and zero elsewhere. Since $J_a$ and $J_b$ are
disjoint, their coordinate-wise product is again zero on all the
deleted blocks. Thus, by \cref{eq:higher-weighted-congruences},
\begin{equation}\label{eq:punctured-weighted-pair}
 \wt_{\nu_i'}(x\star x')
 =\wt_{\nu_i}(\widehat x\star\widehat x')
 \equiv0\bmod2^{h-1}.
\end{equation}

Similarly, for the overlap of $r$ distinct rows, with $3\leq r\leq h$, take their binary lifts $\widehat x^{(1)},\ldots,\widehat x^{(r)}\in\mathcal D_i$. The product of any two of these lifts is zero on the deleted blocks, as we have just proved, so the product of all $r$ lifts is also zero there. Consequently, \cref{eq:higher-weighted-congruences} gives
\begin{equation}\label{eq:punctured-weighted-overlap}
 \wt_{\nu_i'}(x^{(1)}\star\cdots\star x^{(r)})
 =\wt_{\nu_i}(\widehat x^{(1)}\star\cdots\star\widehat x^{(r)})
 \equiv0\bmod2^{h-r+1}.
\end{equation}
This proves the required weighted overlap conditions.

Since every coordinate of $\nu_i'$ is odd, reduction modulo $2$ in \cref{eq:logical-weight-after-puncture,eq:stabilizer-weight-after-puncture,eq:punctured-weighted-pair,eq:punctured-weighted-overlap} shows that $\wt(g^{(a)})\equiv1\bmod2,$ $\wt(v)\equiv0\bmod2$, for $v\in C_{2,i}$. Every overlap of $r$ distinct rows of $G_i$, for $2\leq r\leq 2^{h-1}$, has even weight. Indeed, we have just proven that the overlap of the corresponding lifts is zero on the deleted blocks, so the weight of the overlap is the same before and after the lift. Since $\mathcal{D}_i^{\star (2^{h-1}-1)}\subset \mathcal{D}_i^\perp$, this overlap is even (we can repeat factors if $r<2^{h-1}$). These rows therefore form a $2^{h-1}$-orthogonal matrix. They are also independent. Indeed, if
\(
 \sum_{a=1}^{K_i}\xi_ag^{(a)}\in C_{2,i}\), for some $\xi_a\in\F_2$, then taking the inner product with $g^{(b)}$ gives $\xi_b=g^{(b)}\cdot\sum_{a=1}^{K_i}\xi_ag^{(a)}=0$ because $g^{(b)}\cdot g^{(b)}=1$, while $g^{(b)}$ is orthogonal to $C_{2,i}$ and to every $g^{(a)}$ with $a\neq b$. Thus $G_i$ is full rank and $\dim_{\F_2}(C_{1,i}/C_{2,i})=K_i$. 

The code $\cQ_i=\operatorname{CSS}(C_{1,i},C_{2,i})$ has length $N_i=sL_i=s(n_i-P_i)$, which proves \cref{eq:correction-free-length}. By \cref{lem:higher-weighted-phase},
$\bigotimes_{p=1}^{N_i}R_h^{(\nu_i')_p}$ acts logically as
\[
 \overline R_h^{\epsilon_1}\otimes\cdots\otimes
 \overline R_h^{\epsilon_{K_i}}.
\]
If $\epsilon_a=2^h-1$, interchange the labels of the two computational basis states of the corresponding logical qubit. Indeed,
\[
 X R_h^{-1}X=\zeta_{2^h}^{-1}R_h,
\]
so this changes the logical $R_h^{-1}$ into $R_h$ up to a global phase.

It remains to bound the distances. Since
$C_{2,i}=\Phi_i(\cA_{2,i})$ and
$C_{1,i}\subset\Phi_i(\cA_{1,i})$,
\cref{lem:puncturing-distances} gives
\begin{align}
 d_X\bigl(\operatorname{CSS}(C_{1,i},C_{2,i})\bigr)
 &\geq d_{\rm in}\bigl(d(\cC_i^{(\tau)})-P_i\bigr)\notag\\
 &\geq d_{\rm in}\bigl(A_\ell^{(\tau)}n_i-P_i\bigr),
 \notag\\
 d_Z\bigl(\operatorname{CSS}(C_{1,i},C_{2,i})\bigr)
 &\geq d((\cC_i^{(\tau)})^\perp)-P_i\notag\\
 &\geq B_\ell^{(\tau)}n_i-P_i.
 \notag
\end{align}
This proves \cref{eq:correction-free-x,eq:correction-free-z}.
Since $P_i\geq K_i$ and $P_i\leq(2^{h-1}-1)K_i$, division by
\cref{eq:correction-free-length} gives
\cref{eq:correction-free-rate}. Also
$A_\ell^{(\tau)}>B_\ell^{(\tau)}$ and $d_{\rm in}\geq1$, so the smaller
distance bound is \cref{eq:correction-free-z}. The function
$(B_\ell^{(\tau)}-x)/(1-x)$ is decreasing for
$0\leq x<B_\ell^{(\tau)}$, which gives
\cref{eq:correction-free-relative-distance}.
\end{proof}

It is also possible to obtain a similar construction in which we apply the same diagonal operator $R_h$ on each physical qubit by repeating coordinates, as we show in the next result. 

\begin{corollary}\label{cor:correction-free-uniform}
Under the hypotheses of \cref{thm:correction-free-general}, there is an explicit family $\{\widetilde{\cQ}_i\}_i$ of CSS codes with parameters $[[\widetilde N_i,K_i,\widetilde\Delta_i]]_2$, with $2^{h-1}$-orthogonal generator matrices and
\begin{equation}\label{eq:uniform-correction-free-length}
N_i\leq\widetilde N_i\leq(2^h-1)N_i
=(2^h-1)s(n_i-P_i).
\end{equation}
The distance bounds in \cref{eq:correction-free-x,eq:correction-free-z} hold with $\cQ_i$ replaced by $\widetilde{\cQ}_i$. In a suitable choice of logical basis, $R_h^{\otimes\widetilde N_i}$ acts as $\overline R_h^{\otimes K_i}$ up to a global phase, with no subsequent correction. Moreover,
\begin{align}
\liminf_{i\to \infty}\frac{K_i}{\widetilde N_i}
&\geq\frac{\kappa}{(2^h-1)s(1-\kappa)},
\notag\\
\liminf_{i\to \infty}\frac{\widetilde\Delta_i}{\widetilde N_i}
&\geq\frac{B_\ell^{(\tau)}-(2^{h-1}-1)\kappa}
{(2^h-1)s\bigl(1-(2^{h-1}-1)\kappa\bigr)}.
\notag
\end{align}
\end{corollary}
\begin{proof}
For $\nu\in\{1,3,\ldots,2^h-1\}^N$, define
\[
\Rep_\nu:\F_2^N\longrightarrow\F_2^{\sum_{j=1}^N\nu_j},
\;
\Rep_\nu(x_1,\ldots,x_N)
=\bigl(\underbrace{x_1,\ldots,x_1}_{\nu_1\text{ times}},
\ldots,\underbrace{x_N,\ldots,x_N}_{\nu_N\text{ times}}\bigr).
\]
Let $\nu'_i$ be as in the proof of \cref{thm:correction-free-general}. Apply $\Rep_{\nu_i'}$ to both $C_{1,i}$ and $C_{2,i}$ and write
\[
\widetilde C_{r,i}=\Rep_{\nu_i'}(C_{r,i}),
\; r=1,2.
\]
Since $\Rep_{\nu_i'}$ is injective,
$\widetilde{\cQ}_i=\operatorname{CSS}(\widetilde C_{1,i},\widetilde C_{2,i})$
still encodes $K_i$ qubits. Its length is
\[
\widetilde N_i
=\sum_{j=1}^{N_i}(\nu_i')_j
\leq(2^h-1)N_i
=(2^h-1)s(n_i-P_i),
\]
and $\widetilde N_i\geq N_i$, which proves
\cref{eq:uniform-correction-free-length}. Moreover, for every
$x\in\F_2^{N_i}$,
\[
\wt(\Rep_{\nu_i'}(x))=\wt_{\nu_i'}(x),
\;
R_h^{\otimes\widetilde N_i}|\Rep_{\nu_i'}(x)\rangle
=\zeta_{2^h}^{\wt_{\nu_i'}(x)}|\Rep_{\nu_i'}(x)\rangle.
\]
Hence $R_h^{\otimes\widetilde N_i}$ on the repeated coordinates has exactly the same phase as $\bigotimes_{p=1}^{N_i}R_h^{(\nu_i')_p}$ before repetition. Repetition sends each coset of $C_{2,i}$ in $C_{1,i}$ to the corresponding coset of $\widetilde C_{2,i}$ in $\widetilde C_{1,i}$, so the logical action is unchanged. Since every coordinate of $\nu_i'$ is odd, repetition preserves the parity of each row and of each overlap. In particular, the repeated generator matrix is still $2^{h-1}$-orthogonal.

Finally, odd repetition cannot decrease either CSS distance. For the $X$-distance this follows from
\[
\wt(\Rep_{\nu_i'}(x))\geq\wt(x).
\]
For the $Z$-distance, for $z\in \F_2^{\widetilde N_i}$ (which we can write as $z=(z_{j,a})$, $1\leq j \leq N_i$, $1\leq a \leq (\nu'_i)_j$), we define the linear map 
$$
\Sigma_{\nu_i'}:\F_2^{\widetilde N_i}\longrightarrow \F_2^{N_i}, \; \Sigma_{\nu_i'}(z)_j
=\sum_{a=1}^{(\nu_i')_j}z_{j,a}.
$$
For $x\in \F_2^{N_i}$, it satisfies
\(
z\cdot \Rep_{\nu_i'}(x)
=\Sigma_{\nu_i'}(z)\cdot x\) and \(
\wt(\Sigma_{\nu_i'}(z))\leq\wt(z).
\)
Thus
\[
z\in\widetilde C_{2,i}^{\perp}
\setminus\widetilde C_{1,i}^{\perp}
\;\Longrightarrow\;
\Sigma_{\nu_i'}(z)
\in C_{2,i}^{\perp}\setminus C_{1,i}^{\perp},
\]
and the distance bounds follow. The rate and relative-distance bounds follow from $\widetilde N_i\leq(2^h-1)N_i$,
$\widetilde\Delta_i\geq\Delta_i$, and
\cref{eq:correction-free-rate,eq:correction-free-relative-distance}.
\end{proof}

\begin{corollary}\label{thm:higher-correction-free}
For every fixed $h\geq3$, there is an explicit asymptotically good family of binary CSS codes with parameters $[[\widetilde N_i,K_i,\widetilde \Delta_i]]_2$, with $2^{h-1}$-orthogonal generator matrices, and for which the gate $R_h^{\otimes \widetilde N_i}$ acts logically as $\overline R_h^{\otimes K_i}$, up to a global phase, with no subsequent correction.
\end{corollary}
\begin{proof}
Take $j=2^{h-1}$, $\tau=2^{j-1}$, and an even $m\geq2j$, so that $\ell=2^{m/2}\geq2\tau$. A map from \cref{prop:higher-map} has the required identity and has an element of odd image weight by \cref{eq:higher-map-parity}. Apply \cref{thm:correction-free-general} with any $0<\kappa<B_\ell^{(\tau)}/(2^{h-1}-1)$, and then apply \cref{cor:correction-free-uniform}.
\end{proof}

\begin{remark}\label{rem:addressable-rotations}
Applying \cref{thm:higher-correction-free} at level $h+1$, together with the standard hierarchy-lowering argument, gives asymptotically good binary CSS codes for which logical $R_h$ rotations are transversally addressable, without a correction; see \cite[Rem.~1.3]{wills_csst_decreasing_improved} and the references therein, and \cite{kashyap_z_rotations} for a general framework for addressable logical diagonal rotations. The parameters are those of our construction at level $h+1$.
\end{remark}

The constants obtained by applying \cref{cor:correction-free-uniform} with the maps from \cref{prop:higher-map} deteriorate quickly with $h$, but remain positive for every fixed value of $h$. For the $T$ gate, the shorter explicit maps in the next section give better parameters.

\section{Explicit constructions for the \texorpdfstring{$T$}{T} gate}\label{sec:families}
We now specialize \cref{thm:general-asymptotic,thm:correction-free-general} to $h=3$. The two arguments have already been proved; it remains to give the explicit maps and compatible sets, and substitute their parameters. All the computations from this section have been checked with SageMath \cite{sagemath}.

\subsection{The cubic identity and compatible sets}\label{ss:T-identities}
We first consider the symmetric $\F_2$-trilinear form obtained from \cref{eq:higher-map-identity} with $j=3$:
\begin{equation*}\label{eq:theta-general}
 \Theta_{\F_{2^m}}(a,b,c)
 =\Tr\bigl(H_3(a,b,c)\bigr)
 =\Tr\bigl(abc(a+b+c)\bigr).
\end{equation*}
Note that  $\Theta_{\F_{2^m}}$ satisfies the identity
\begin{equation*}\label{eq:theta-collision}
 \Theta_{\F_{2^m}}(a,a,b)=\Theta_{\F_{2^m}}(a,b,b)=\Tr(ab).
\end{equation*}
For an additive map $\phi:\F_{2^m}\to\F_2^s$, consider the $\F_2$-trilinear form
\[
 \mathcal{T}_\phi(a,b,c)=
 \sum_{j=1}^s\phi_j(a)\phi_j(b)\phi_j(c).
\]
For the construction with a Clifford correction, we use injective maps
$\phi$ such that
\begin{align}
 \phi(a)\cdot\phi(b)&=\Tr(ab),
 \label{eq:phi-bilinear}\\
 \mathcal{T}_\phi(a,b,c)&=\Theta_{\F_{2^m}}(a,b,c)
 \label{eq:phi-trilinear}
\end{align}
for all $a,b,c\in\F_{2^m}$.

For $h=3$ and a map $\phi$ satisfying
\cref{eq:phi-bilinear,eq:phi-trilinear}, the conditions in
\cref{def:frame} become
\begin{equation}\label{eq:frame}
\operatorname{Tr}(u_au_b)=\delta_{ab},\text{ for all }a,b,
\;
\Theta_{\F_{2^m}}(u_a,u_b,u_c)=0,\text{ for pairwise distinct }a,b,c.
\end{equation}
In this section, a compatible set always means one satisfying \cref{eq:frame}. For $\tau=4$, write $\cC_i=\cC_i^{(4)}$, $A_\ell=A_\ell^{(4)}$, and $B_\ell=B_\ell^{(4)}$. Finding a map whose image is a code of small length (small $s$) or large minimum distance can improve the resulting parameters of \cref{thm:general-asymptotic,thm:correction-free-general}. We now give the explicit maps used for the $T$-gate constructions.

\subsection{A construction over \texorpdfstring{$\F_{64}$}{F64}}\label{ss:explicit-maps}

We first give the shorter map over $\F_{64}$. For $m=6$, write
\[
 \F_{64}=\F_2[\alpha]/(\alpha^6+\alpha+1),
 \;
 a=x_0+x_1\alpha+\cdots+x_5\alpha^5.
\]

\begin{proposition}\label{prop:inner64}
Let $M\in\F_2^{6\times14}$ be
\[
\renewcommand{\arraystretch}{0.92}
M=\left(
\begin{array}{rrrrrrrrrrrrrr}
0&0&1&1&1&1&0&1&0&0&1&0&0&0\\
1&0&0&0&1&1&0&0&0&0&0&1&0&0\\
0&0&1&0&1&0&1&1&0&1&1&1&0&1\\
0&1&1&1&1&0&0&0&1&1&1&1&0&0\\
0&0&0&1&1&0&0&0&0&0&0&0&1&1\\
0&0&0&0&0&1&1&1&1&1&1&1&1&1
\end{array}\right).
\]
Define $\phi_{64}:\F_{64}\to\F_2^{14}$ by $\phi_{64}(a)=(x_0,\ldots,x_5)M$. Then $\mathcal{T}_{\phi_{64}}=\Theta_{\F_{64}}$, and its image is a binary $[14,6,4]$ code with weight enumerator
\begin{equation}\label{eq:weight-enumerator64}
 1+3z^4+10z^5+11z^6+12z^7+12z^8+10z^9+5z^{10}.
\end{equation}
Moreover, the following is a compatible set of size four:
\begin{align*}
 u_1&=\alpha^5,\\
 u_2&=1+\alpha^3+\alpha^5,\\
 u_3&=1+\alpha^2+\alpha^4+\alpha^5,\\
 u_4&=1+\alpha+\alpha^2+\alpha^4+\alpha^5.
\end{align*}
\end{proposition}

\begin{proof}
Both the identities in \cref{eq:phi-bilinear,eq:phi-trilinear} are multilinear over $\F_2$, so it is enough to evaluate them on $1,\alpha,\ldots,\alpha^5$, reducing products by $\alpha^6=\alpha+1$. Enumerating directly gives \cref{eq:weight-enumerator64}. The same finite calculation gives
\[
 \bigl(\Tr(u_au_b)\bigr)_{a,b=1}^4=I_4
\]
and $\Theta_{\F_{64}}(u_a,u_b,u_c)=0$ for the four triples of distinct elements in the compatible set.
\end{proof}

For $\ell=8$, the bounds in \cref{prop:outer} are respectively $5/28$, $19/28$, and $1/28$.

\begin{corollary}\label{thm:F64}
Fix $0<\kappa<1/28$ and put $P_i=\lfloor\kappa n_i\rfloor$. There is an explicit family $\{\cQ_i\}_i$ with parameters $[[N_i,K_i,\Delta_i]]_2$, $N_i=14(n_i-P_i)$, $K_i=4P_i$, for which
\begin{align*}
 d_X(\cQ_i)&\geq4\left(\frac{19}{28}n_i-P_i\right),\\
 d_Z(\cQ_i)&\geq\frac1{28}n_i-P_i,
\end{align*}
and
\begin{align*}
 \liminf_{i\to \infty}\frac{K_i}{N_i}
 &\geq\frac{4\kappa}{14(1-\kappa)},\\
 \liminf_{i\to \infty}\frac{\Delta_i}{N_i}
 &\geq\frac{1/28-\kappa}{14(1-\kappa)}.
\end{align*}
For each $i$, a diagonal Clifford correction $U_i$ can be computed from the defining triorthogonal matrix of $\cQ_i$ such that $U_iT^{\otimes N_i}$ acts logically as $\overline T^{\otimes K_i}$.
\end{corollary}

\begin{proof}
Apply \cref{thm:general-asymptotic} with $h=3$, $\tau=4$, $m=6$, $\ell=8$, $s=14$, $d_{\rm in}=4$, $t=4$, using the inner map and compatible set from \cref{prop:inner64}.
\end{proof}

\subsection{A construction over \texorpdfstring{$\F_{256}$}{F256}}\label{ss:explicit256}

For the second map, take $m=8$ and write
\[
 \F_{256}
 =\F_2[\alpha]/(\alpha^8+\alpha^4+\alpha^3+\alpha+1).
\]
For \(v=(v_0,\ldots,v_7)\in\F_2^8\), define $\ell_v\left(\sum_{i=0}^7 a_i\alpha^i\right)=\sum_{i=0}^7 v_i a_i$. We represent \(v\) by the binary string \(v_7\cdots v_0\).

\begin{proposition}\label{prop:inner256}
Let \(\phi_{256}:\F_{256}\to\F_2^{25}\) have as its coordinate
functionals the \(\ell_v\) indexed by the following binary strings
\(v_7\cdots v_0\):
\[
\begin{array}{rrrrr}
00000001&00000011&00000110&00010000&00101010\\
00110011&01000100&01011101&01011110&01100000\\
01101110&01110010&01110101&01111000&01111010\\
01111100&10000011&10000111&10001111&10010000\\
10010111&10101111&10110100&11000111&11001101.
\end{array}
\]
Then $\mathcal{T}_{\phi_{256}}=\Theta_{\F_{256}}$, and its image is a binary $[25,8,7]$
code with weight enumerator
\begin{align}
1&+z^7+8z^8+24z^9+19z^{10}+31z^{11}+48z^{12}+32z^{13}\notag\\
 &+34z^{14}+31z^{15}+15z^{16}+8z^{17}+3z^{18}+z^{19}.
\end{align}
The five field elements \(u=\sum_{i=0}^7u_i\alpha^i\) whose binary
coefficient strings \(u_7\cdots u_0\) are
\begin{equation}\label{eq:frame256}
11011111,\quad11011010,\quad11010000,\quad
11000110,\quad01110011
\end{equation}
form a compatible set. 
\end{proposition}

\begin{proof}
The trilinear identity is checked on the $8^3$ triples of polynomial-basis
elements. Enumerating the $256$ inputs gives the stated weight enumerator.
The $5\times5$ trace-pairing matrix of the elements in
\cref{eq:frame256} is the identity, and the cubic tensor vanishes on all ten
triples of distinct elements.
\end{proof}

For $\ell=16$, the bounds in \cref{prop:outer} are respectively $13/60$, $43/60$, and $3/20$. The substantial increase in the last quantity is the principal source of the improved asymptotic constant over $\F_{256}$ that we obtain below.

\begin{corollary}\label{thm:F256}
Fix $0<\kappa<3/20$ and put $P_i=\lfloor\kappa n_i\rfloor$. There is an explicit family $\{\cQ_i\}_i$, with parameters $[[N_i,K_i,\Delta_i]]_2$, $N_i=25(n_i-P_i)$, $K_i=5P_i$, for which
\begin{align*}
 d_X(\cQ_i)&\geq7\left(\frac{43}{60}n_i-P_i\right),\\
 d_Z(\cQ_i)&\geq\frac3{20}n_i-P_i,
\end{align*}
and
\begin{align*}
 \liminf_{i\to \infty}\frac{K_i}{N_i}
 &\geq\frac{5\kappa}{25(1-\kappa)},\\
 \liminf_{i\to \infty}\frac{\Delta_i}{N_i}
 &\geq\frac{3/20-\kappa}{25(1-\kappa)}.
\end{align*}
For each $i$, a diagonal Clifford correction $U_i$ can be computed from the defining triorthogonal matrix of $\cQ_i$ such that $U_iT^{\otimes N_i}$ acts logically as $\overline T^{\otimes K_i}$.
\end{corollary}

\begin{proof}
Apply \cref{thm:general-asymptotic} with $h=3$, $\tau=4$, and $m=8$, $\ell=16$, $s=25$, $d_{\rm in}=7$, $t=5$, using the inner map and compatible set from
\cref{prop:inner256}.
\end{proof}

\subsection{Further extension fields and limitations}\label{ss:further-fields}
For a fixed extension field and fixed $s,d_{\rm in},t$, the
construction gives
\[
 \mathcal{R}=\frac{t\kappa}{s(1-\kappa)},
 \;
 \delta_{\rm lb}
 =\frac{B_\ell-\kappa}{s(1-\kappa)},
 \; 0<\kappa<B_\ell.
\]
Eliminating
$\kappa$ gives the rate--distance tradeoff
\begin{equation*}\label{eq:general-tradeoff}
 \delta_{\rm lb}(\mathcal{R})
 =
 \frac{B_\ell}{s}-\frac{1-B_\ell}{t}\mathcal{R},
 \;
 0<\mathcal{R}<\frac{tB_\ell}{s(1-B_\ell)}.
\end{equation*}
For example, for the construction over $\F_{64}$, we obtain \[
\delta_{\rm lb}(\mathcal{R})
=\frac1{392}-\frac{27}{112}\mathcal{R},
\;
0<\mathcal{R}<\frac2{189},
\]
while for the construction over $\F_{256}$, we obtain
\[
\delta_{\rm lb}(\mathcal{R})
=\frac3{500}-\frac{17}{100}\mathcal{R},
\;
0<\mathcal{R}<\frac3{85}.
\]
Over the entire range of the $\F_{64}$ construction, the $\F_{256}$ construction has better relative parameters (in terms of the bounds we have obtained). However, it remains unclear whether increasing the extension degree leads to better trade-offs, since the parameters also depend on $s$ and $t$.

\subsection{A correction-free construction over \texorpdfstring{$\F_{256}$}{F256}}

For an additive map
$\psi:\F_{2^m}\to\F_2^s$, define
\[
 \sigma_\psi(a,b,c,d)
 =\sum_{j=1}^s\psi_j(a)\psi_j(b)\psi_j(c)\psi_j(d).
\]

In the correction-free construction, we use maps $\psi$ satisfying the general identity in \cref{eq:higher-map-identity} with $j=4$:
\begin{equation}\label{eq:fourth-map}
 \sigma_\psi(a,b,c,d)=\Tr\bigl(H_4(a,b,c,d)\bigr)
\end{equation}
for $a,b,c,d\in\F_{2^m}$. Since $H_4$ is homogeneous of total degree $8$, \cref{lem:schur-transfer} applies with $r=4$ and $\tau=8$. By \cref{eq:canonical-overlap-collision} and invariance of the trace under squaring, the same argument as in the proof of \cref{prop:higher-map} gives
\begin{align}
 \psi(a)\cdot\psi(b)
 &=\Tr(ab),
 \label{eq:fourth-pairing}\\
 \wt(\psi(a))
 &\equiv\Tr(a)\pmod2.
 \notag
\end{align}
In particular, $\psi$ is injective, and exactly half of the
field elements have odd image weight.

We use the representation of $\F_{256}$ in the basis $\{1,\alpha,\ldots,\alpha^7\}$ and the functionals $\ell_v$ fixed in \cref{ss:explicit256}.

\begin{proposition}\label{prop:fourth-map-256}
Let $\psi_{256}:\F_{256}\to\F_2^{93}$ have as its coordinate functionals the $\ell_v$ indexed by the following binary coefficient strings $v_7\cdots v_0$:
\[
\renewcommand{\arraystretch}{0.94}
\begin{array}{rrrrrr}
00110100&10111011&10001111&01111100&01000011&01110111\\
10000100&00011110&00101010&11001100&10100101&10000111\\
10010010&00111100&11001111&11111011&10101101&11100101\\
01111111&00010110&10001100&10111000&11111100&11001000\\
00111011&10100001&10010101&10000000&00101110&11100010\\
01111000&01001100&10011110&11000011&11010101&01001111\\
00010010&00100110&11000000&10011101&11001011&01110000\\
11111001&00111110&01000010&10110001&10001110&00100000\\
11010011&11100111&01111101&10101111&00000001&01101000\\
11110010&11000110&11100100&00100011&00010111&11000101\\
01101011&10011000&10010011&11111010&01010100&00011100\\
00110001&01011000&00010000&01001101&01111001&11011100\\
10110101&10000001&00011011&01010011&11001001&00001110\\
10100000&00111010&10110110&11011111&01110001&00111001\\
10100011&10010111&01010000&11000001&11110101&10011100\\
00000110&10101000&01111010&&&
\end{array}
\]
Then $\psi_{256}$ satisfies \cref{eq:fourth-map}. Its image is a binary $[93,8,37]$ code, and $\wt(\psi_{256}(\alpha^5))=45$.
\end{proposition}

\begin{proof}
The identity in \cref{eq:fourth-map} is four-linear over $\F_2$, so it is enough to check the $8^4$ quadruples of polynomial-basis elements. Enumerating the $255$ nonzero inputs gives minimum weight $37$, and direct evaluation at $\alpha^5$ gives weight $45$.
\end{proof}

For $\ell=16$, the constants in \cref{prop:outer}, with $\tau=8$, are
\[
 \rho_{16}^{(8)}=\frac3{40},\;
 A_{16}^{(8)}=\frac{103}{120},\;
 B_{16}^{(8)}=\frac1{120}.
\]
Thus \cref{thm:correction-free-general,cor:correction-free-uniform} give the following family.

\begin{corollary}\label{cor:F256-correction-free}
Fix $0<\kappa<1/360$. There is an explicit family $\{\widetilde{\cQ}_i\}_i$ of CSS codes with parameters $[[\widetilde N_i,K_i,\widetilde\Delta_i]]_2$, $K_i=\lfloor\kappa n_i\rfloor$, with $4$-orthogonal generator matrices, and $K_i\leq P_i\leq3K_i$ such that
\begin{align*}
 \widetilde N_i&\leq651(n_i-P_i),\\
 d_X(\widetilde{\cQ}_i)
 &\geq37\left(\frac{103}{120}n_i-P_i\right),\\
 d_Z(\widetilde{\cQ}_i)
 &\geq\frac1{120}n_i-P_i.
\end{align*}
Then $T^{\otimes\widetilde N_i}$ acts logically as $\overline T^{\otimes K_i}$ in the logical basis fixed in \cref{thm:correction-free-general}, up to a global phase and without a Clifford correction, and
\begin{align*}
 \liminf_{i\to \infty}\frac{K_i}{\widetilde N_i}
 &\geq\frac{\kappa}{651(1-\kappa)},\\
 \liminf_{i\to \infty}\frac{\widetilde\Delta_i}{\widetilde N_i}
 &\geq\frac{1/120-3\kappa}{651(1-3\kappa)}.
\end{align*}
\end{corollary}
\begin{proof}
Apply \cref{thm:correction-free-general} with $h=3$, $\tau=8$, $s=93$, $d_{\rm in}=37$, and $u=\alpha^5$, using \cref{prop:fourth-map-256}, and then apply \cref{cor:correction-free-uniform}.
\end{proof}

\subsection{Constant-overhead \texorpdfstring{$T$}{T}-state distillation}
\label{ss:constant-overhead}
In this section, we assume ideal stabilizer operations. Let $\ket{T}=T\ket{+}$ and $\tau_T=\ket{T}\bra{T}$. This exposition follows an argument similar to the one in the supplementary material of \cite{wills_constant_overhead_msd}. After independent Clifford twirling \cite{bravyiTriorthogonalOriginal}, the inputs are a mixture of $Z^e\ket{T}^{\otimes N}$, with $e\in\F_2^N$. We assume local stochastic noise of strength $p$, meaning that
\[
 \Pr(F\subset\supp(e))\leq p^{|F|}
\]
for every $F\subset\{1,\ldots,N\}$. This includes independent input errors of probability $p$. The following corollary applies the correction-based argument of \cite{wills_constant_overhead_msd} to the binary triorthogonal distillation protocol of \cite{bravyiTriorthogonalOriginal}.

\begin{corollary}\label{cor:constant-overhead-T}
Consider either family in \cref{thm:F64,thm:F256}, with a fixed admissible $\kappa$. There is $p_0>0$ such that, for every fixed $0<p<p_0$ and all sufficiently large $i$, a protocol using ideal stabilizer operations always returns $K_i$ states from $N_i$ noisy $T$ states, whose joint output state $\rho_i$ satisfies
\begin{equation*}\label{eq:trace_distance}
 \frac12\bigl\|\rho_i-\tau_T^{\otimes K_i}\bigr\|_1
 \leq e^{-cN_i}
\end{equation*}
for some $c>0$ independent of $i$. The input cost per output is exactly $N_i/K_i=O(1)$. Thus these families give direct constant-overhead $T$-state distillation.
\end{corollary}

\begin{proof}
Prepare encoded $\ket{+}^{\otimes K_i}$ and inject the physical $T^{\otimes N_i}$ using the input states. An input error $e$ gives $Z^eT^{\otimes N_i}$. Apply the diagonal Clifford correction $U_i$ from \cref{lem:higher-correction}. Since $U_i$ commutes with $Z^e$, the resulting state is $Z^e\ket{\overline T^{\otimes K_i}}$. Correct the $Z$ errors using minimum-weight syndrome decoding, and decode the corrected block to obtain the $K_i$ output qubits.

Write $d_{Z,i}=d_Z(\cQ_i)$. Minimum-weight recovery corrects every error of weight at most $r_i=\lfloor(d_{Z,i}-1)/2\rfloor$. Our distance bounds give $d_{Z,i}=\Omega(N_i)$, so there is $a>0$ such that $w_i:=r_i+1\geq aN_i$ for all sufficiently large $i$.
For $p<a/\mathrm e$,
\[
 \mathfrak{q}_i:=\Pr(\wt(e)\geq w_i)
 \leq\binom{N_i}{w_i}p^{w_i}
 \leq(\mathrm ep/a)^{aN_i}.
\]
On successful recovery, the output is exactly $\tau_T^{\otimes K_i}$, so this probability bounds the error of all $K_i$ outputs jointly. The output state can be written as  $\rho_i=(1-\mathfrak{q}_i)\tau_T^{\otimes K_i}+\mathfrak{q}_i\sigma_i$ for some state $\sigma_i$. Then 
$$
\frac12\bigl\|\rho_i-\tau_T^{\otimes K_i}\bigr\|_1
\leq \mathfrak{q}_i
\leq \left(\frac{\mathrm ep}{a}\right)^{aN_i}.
$$
Take $p_0=a/\mathrm e$ and $c=a\log(a/(\mathrm ep))$. The fact that $K_i=\Theta(N_i)$ gives the claimed input cost.
\end{proof}

\begin{remark}
The AG decoder in \cref{app:z-decoding} also gives a linear guaranteed $Z$-error correction radius and runs in polynomial time. For $h>3$, the lower-level corrections, including those arising from state injection, are generally non-Clifford. The same distillation argument applies if operations from $\mathfrak C_{h-1}$ are assumed ideal and are not counted in the overhead.
\end{remark}

\section{AI Usage}
The explicit inner maps from \cref{sec:families} were found with ChatGPT 5.6 Sol. This AI was also used to help organize drafts, clarify proofs, and write code to check the results. All content, claims, and conclusions have been reviewed and verified by the author, and they remain the responsibility of the author.

\appendix

\section{The general alphabet-reduction map}\label{app:general-map}

We first prove the properties of the polynomials in \cref{eq:canonical-overlap-product}, and then construct the map $\phi$ by solving a triangular system over $\F_2$.

\begin{lemma}\label{lem:overlap-polynomial}
For $j\geq1$, the polynomial $H_j$ is symmetric, homogeneous of degree $2^{j-1}$, and additive in each variable in characteristic two. Moreover,
\begin{equation}\label{eq:canonical-overlap-collision}
 H_j(X_1,\ldots,X_{j-2},Y,Y)
 =H_{j-1}(X_1,\ldots,X_{j-2},Y)^2,
\end{equation}
for $j\geq 2$. 
\end{lemma}

\begin{proof}
Symmetry and homogeneity follow from the definition. For a finite binary subspace $W$ of a field of characteristic two, the polynomial $L_W(Z)=\prod_{w\in W}(Z+w)$ is additive. Indeed, $L_{\{0\}}(Z)=Z$, and, if $u\notin W$, induction gives
\[
 L_{W\oplus\langle u\rangle}(Z)
 =L_W(Z)L_W(Z+u)
 =L_W(Z)^2+L_W(u)L_W(Z),
\]
which is again additive. For $j\geq2$, work over $\F_2[X_1,\ldots,X_{j-1}]$ and take
\[
 W=\left\{\sum_{a\in S}X_a:
 S\subset\{1,\ldots,j-1\},\ |S|\text{ even}\right\}=\Span_{\F_2}\{X_1+X_{j-1},\dots,X_{j-2}+X_{j-1}\}.
\]
Separating the odd subsets according to whether they contain $j$ gives
\[
 H_j(X_1,\ldots,X_j)
 =H_{j-1}(X_1,\ldots,X_{j-1})L_W(X_j).
\]
Thus $H_j$ is additive in its last variable, and hence in every variable by symmetry. Also note that $H_1(X)=X$ is immediate.

Finally, the map $S\mapsto S\triangle\{j-1\}$ is a bijection between the even and odd subsets of $\{1,\ldots,j-1\}$. Hence
\[
 L_W(X_{j-1})
 =\prod_{\substack{S\subset\{1,\ldots,j-1\}\\ |S|\text{ even}}}
 \left(X_{j-1}+\sum_{a\in S}X_a\right)
 =H_{j-1}(X_1,\ldots,X_{j-1}).
\]
Substituting $X_j=X_{j-1}$ into the preceding factorization gives
\[
 H_j(X_1,\ldots,X_{j-2},X_{j-1},X_{j-1})
 =H_{j-1}(X_1,\ldots,X_{j-1})^2,
\]
which proves \cref{eq:canonical-overlap-collision}.
\end{proof}

\begin{proof}[Proof of \cref{prop:higher-map}]
For $1\leq r\leq j$, consider
\[
\Gamma_r(x_1,\ldots,x_r)
:=\Tr\bigl(H_r(x_1,\ldots,x_r)\bigr).
\]
These forms are symmetric and multilinear over $\F_2$. By
\cref{eq:canonical-overlap-collision} and invariance of the trace
under squaring, they satisfy
\[
\Gamma_r(x_1,\ldots,x_{r-2},y,y)
=\Gamma_{r-1}(x_1,\ldots,x_{r-2},y),
\]
for $2\leq r \leq j$. Choose a binary basis $\beta_1,\ldots,\beta_m$ of $\F_{2^m}$, and let
\[
\mathcal I=\{I\subset\{1,\ldots,m\}:1\leq|I|\leq j\}.
\]
For $I\in\mathcal I$, put $t_I=\Gamma_{|I|}((\beta_i)_{i\in I})$, where the order is irrelevant by symmetry. Choose $c_U\in\F_2$, $U\in\mathcal I$, satisfying
\begin{equation}\label{eq:binary-moment-system}
\sum_{\substack{U\in\mathcal I\\ I\subset U}}c_U=t_I,
\end{equation}
for $I\in\mathcal I$. This system has a unique solution: ordering subsets by decreasing cardinality makes its coefficient matrix triangular with diagonal entries equal to one. Indeed, put $d=\min\{j,m\}$. For every $I$ with $|I|=d$, the equation gives $c_I=t_I$. For $|I|=d-1$, it gives
\[
c_I+\sum_{\substack{I\subsetneq U\\ |U|=d}}c_U=t_I,
\]
which determines $c_I$, since all coefficients in the sum have already been determined. Continuing by decreasing cardinality determines all the coefficients uniquely.
For $U\in\mathcal I$, define the binary linear functional
\[
\ell_U\!\left(\sum_{i=1}^m a_i\beta_i\right)=\sum_{i\in U}a_i,
\]
where $a_i \in \F_2$ for $1\leq i \leq m$, and set $\phi(x)=(\ell_U(x))_{U\in\mathcal I:\,c_U=1}$. Note that $\ell_U(\beta_i)=1$ if $i\in U$, and it is 0 otherwise. For any tuple of basis elements $\beta_{i_1},\ldots,\beta_{i_r}$, with $1\leq r\leq j$, let $I=\{i_1,\ldots,i_r\}$. Its binary overlap is
\[
\wt\!\left(
\mathop{\bigstar}_{a=1}^r\phi(\beta_{i_a})
\right)\bmod 2= \sum_{U\in\mathcal I}c_U\prod_{a=1}^r\ell_U(\beta_{i_a})
=\sum_{I\subset U}c_U=t_I
=\Gamma_r(\beta_{i_1},\ldots,\beta_{i_r}),
\]
where the last equality follows by repeatedly removing a repeated argument using the collision identity. Multilinearity extends the equality to all inputs, proving \cref{eq:higher-map-identity} and, for $r=1$, \cref{eq:higher-map-parity}. For $r=2$, the same identity gives
\[
\phi(x)\cdot\phi(y)=\Tr(xy).
\]
Since this pairing is nondegenerate, $\phi$ is injective. Finally,
\[
s\leq|\mathcal I|=\sum_{a=1}^{\min\{j,m\}}\binom ma.
\]
\end{proof}

\begin{remark}
The construction uses only finite-field arithmetic and back substitution in \cref{eq:binary-moment-system}, so it is explicit.
\end{remark}

\section{Higher-level odd completion}\label{app:higher-completion}

\begin{proof}[Proof of \cref{lem:higher-odd-completion}]
We argue by induction on $h$. The case $h=2$ is a special case of \cite[Lem.~III.4]{haah_towers_divisible_codes}. We give a direct proof of this case and then argue by induction on $h$. Thus, let $h=2$. The hypothesis is
$D\subset D^\perp$. Hence
\[
 \lambda(x)=\frac{\wt(x)}2\bmod2
\]
is linear on $D$. Extend it to $\F_2^N$, represent the extension by
$s\in\F_2^N$, and take $\nu=\one-2s\bmod4$, with every coordinate
represented by $1$ or $3$. Then
\[
 \wt_\nu(x)=0\bmod4,
 \;
 \wt_\nu(x\star y)=0\bmod2
\]
for $x,y\in D$.

Suppose now that the result holds at level $h-1$, and put
$P=D^{\star2}$. We have
\[
 P^{\star(2^{h-2}-1)}\subset P^\perp.
\]
Indeed, to test orthogonality to $P$, it is enough to take products generating $P^{\star(2^{h-2}-1)}$ and $P$. Their inner product is the sum of a product of $2^{h-1}$ words of $D$, and is therefore zero. By induction, there is a vector with odd entries $\nu^{(0)}$ modulo $2^{h-1}$ satisfying the required congruences on $P$. In particular,
\begin{equation}\label{eq:higher-completion-pairs}
 \wt_{\nu^{(0)}}(x\star y)=0\bmod2^{h-1},
\end{equation}
for $x,y\in D$. We prove the remaining congruences by induction on the number of codewords \(r\), for \(3\le r\le h\). Suppose the claim holds for \(r-1\), and consider
\[
 z=x^{(1)}\star x^{(3)}\star\cdots\star x^{(r)},
 \;
 z'=x^{(2)}\star x^{(3)}\star\cdots\star x^{(r)}.
\]
Then $z\star z'=x^{(1)}\star \cdots \star x^{(r)}$ and $z+z'=(x^{(1)}+x^{(2)})\star x^{(3)}\cdots \star x^{(r)}$ (in particular, $z+z'$ is a product of $r-1$ codewords). If we consider 
\[
 \wt_{\nu^{(0)}}(z+z')
 =\wt_{\nu^{(0)}}(z)+\wt_{\nu^{(0)}}(z')
 -2\wt_{\nu^{(0)}}(z\star z'),
\]
since $\wt_{\nu^{(0)}}(z+z')\equiv \wt_{\nu^{(0)}}(z)\equiv \wt_{\nu^{(0)}}(z')\equiv 0\bmod 2^{h-r+2}$, we obtain 
$$
\wt_{\nu^{(0)}}(x^{(1)}\star \cdots \star x^{(r)})=\wt_{\nu^{(0)}}(z\star z')\equiv0\bmod 2^{h-r+1}.
$$

Taking $y=x$ in \cref{eq:higher-completion-pairs} shows that
$\wt_{\nu^{(0)}}(x)$ is divisible by $2^{h-1}$. Moreover,
\[
 \mu(x)=\frac{\wt_{\nu^{(0)}}(x)}{2^{h-1}}\bmod2
\]
is linear on $D$, again by \cref{eq:higher-completion-pairs} and the weight
identity. Extend $\mu$ to $\F_2^N$, represent it by $s\in\F_2^N$, and set
\[
 \nu=\nu^{(0)}-2^{h-1}s\bmod2^h.
\]
Every coordinate of $\nu$ is odd, and the new term makes the weighted sum of
every word of $D$ zero modulo $2^h$. For a product of $r\geq2$ words, the
change is a multiple of $2^{h-1}$ and therefore preserves the congruence
modulo $2^{h-r+1}$. This completes the induction.
\end{proof}

\section{\texorpdfstring{$Z$}{Z}-error correction using the trace adjoint}
\label{app:z-decoding}
The map $\rho_\phi$ from \cref{eq:adjoint}, used to bound $d_Z$ in \cref{lem:puncturing-distances}, also allows us to reduce the $Z$-error-correction step in distillation to classical syndrome decoding.

\begin{lemma}\label{lem:trace-adjoint-decoding}
Let $\phi:\F_{2^m}\to\F_2^s$ be additive and injective, and consider $I=\phi(\F_{2^m})$ and 
\[
 \Phi:\F_{2^m}^L\longrightarrow\F_2^{sL}, \; \Phi(a_1,\dots,a_L)=(\phi(a_1),\dots,\phi(a_L)). 
\]
Let $\cA_2\subset\F_{2^m}^L$ be $\F_{2^m}$-linear and suppose that
\[
 C_2=\Phi(\cA_2)\subset C_1\subset I^L.
\]
A decoder for $\cA_2^\perp$ correcting every error of symbol weight at most $r$ induces a $Z$ decoder for $\operatorname{CSS}(C_1,C_2)$ correcting every binary error of weight at most $r$. 
\end{lemma}

\begin{proof}
By \cref{eq:adjoint},
\[
 e\cdot\phi(x)=\Tr(\rho_\phi(e)x).
\]
The kernel of $\rho_\phi$ is $I^\perp$, of binary dimension $s-m$, and thus, its image has dimension $m$, i.e., $\rho_\phi$ is surjective. Fix an $\F_2$-linear map $\sigma:\F_{2^m}\to\F_2^s$ such that $\rho_\phi\circ \sigma = \text{id}_{\F_{2^m}}$ (this can be done by choosing one preimage for each element of a basis of $\F_{2^m}$ over $\F_2$).

Let $v_1,\ldots,v_a$ be an $\F_{2^m}$-basis of $\cA_2$, and let $\beta_1,\ldots,\beta_m$ and $\beta_1^*,\ldots,\beta_m^*$ be trace-dual binary bases of $\F_{2^m}$. The vectors $\Phi(\beta_jv_b)$, $1\leq j \leq m$, $1\leq b \leq a$, form an $\F_2$-basis of $C_2$. For an unknown binary error $e$, put $z=\rho_{\phi,L}(e)$. Because of our conventions for the CSS construction, the corresponding syndrome is obtained by multiplying with basis elements of $C_2$:
\[
 s_{b,j}=e\cdot\Phi(\beta_jv_b)= \Tr\left( \beta_j\sum_{c=1}^L (v_b)_c\rho_{\phi,L}(e)_c \right)
 =\Tr(\beta_jh_b)
\]
where $h_b=\sum_{c=1}^L(v_b)_cz_c$. Hence $h_b=\sum_{j=1}^m s_{b,j}\beta_j^*$. These are the syndrome coordinates over $\F_{2^m}$ of $z$ for the parity-check matrix $H$ whose rows are $v_b$.

Solve $Hy^{\mathsf T}=h$ for any $y\in\F_{2^m}^L$. Then $y-z\in\cA_2^\perp$, so a decoder for $\cA_2^\perp$ applied to $y$ recovers $z$ whenever $\wt_q(z)\leq r$ (we could have directly applied a syndrome decoder instead without looking for $y$, but this will be relevant for the next proposition). Every nonzero coordinate of $z$ comes from a nonzero binary block of $e$, and therefore
\(
 \wt_q(z)\leq\wt_2(e).
\)
Now we have
\[
 e+(\sigma(z_1),\dots,\sigma(z_L))\in
 \ker\rho_{\phi,L}=(I^\perp)^L\subset C_1^\perp.
\]
This is a $Z$ stabilizer in our convention, and thus, we recovered the error (up to stabilizers).
\end{proof}

\begin{remark}
For the two inner maps used in \cref{thm:F64,thm:F256}, \cref{eq:phi-bilinear} gives $\rho_\phi(\phi(x))=x$. Thus we may take $\sigma=\phi$. 
\end{remark}

By \cref{lem:trace-adjoint-decoding}, we reduced the problem of correcting $Z$ errors to decoding $\mathcal{A}_2^\perp$. Since our construction uses AG codes, we can use standard AG decoders for this, as we show next. We use the notation introduced in \cref{sec:outer,ss:ag_codes}.

\begin{proposition}\label{prop:AG-Z-decoder}
Consider the construction in \cref{thm:general-puncture} using an outer code $C_L(D,G)$ from \cref{prop:outer}. Let $S$ be the divisor of the $P\geq1$ deleted information coordinates and put $D'=D-S$, $L=n-P$, and $R=(du_0/u_0)$. If $d_*=\deg G-P+2-2g>0$, the corresponding binary code has a $Z$ decoder correcting every error of weight at most
\(
 \left\lfloor\frac{d_*-1}{2}\right\rfloor.
\)
For fixed $q$ and inner map, its complexity is $O(N^3)$.
\end{proposition}

\begin{proof}
Shortening on $S$ gives
\[
 \cA_2=C_L(D',G-S).
\]
Consequently, we have
\[
 \cA_2^\perp=C_\Omega(D',G-S),
\]
which is the differential AG code associated to $D'$ and $G-S$. We apply the general AG decoder of \cite[Thm. 3]{farran_decoding_ag}, which has complexity $O(L^{2.81})$. It requires an extra rational point $Q$ outside of $D'$, which we can take from $S$. It decodes any error with weight at most $\left\lfloor \frac{\deg(G-S)+2-2g-1}{2}\right\rfloor $. Since we have
$$
\deg(G-S)+2-2g=\deg G-P+2-2g=d_*,
$$
it decodes any error with weight at most $\left \lfloor \frac{ d_*-1}{2} \right\rfloor$. With respect to the rest of the hypotheses of \cite[Thm. 3]{farran_decoding_ag}, just note that $d_*>0\iff \deg(G-S)>2g-2$ and $\deg(G-S)<L$, which follows from $\deg G<n$. The binary recovery statement follows from \cref{lem:trace-adjoint-decoding}, and the total complexity is $O(N^3)$ if we need to compute a right inverse of $H$ in the proof of \cref{lem:trace-adjoint-decoding} to obtain a solution $y$ (it would be $O(N^{2.81})$ if this is precomputed). 
\end{proof}

\begin{remark}
For the asymptotic families considered above, the choice of $P$ gives $d_*\geq B_\ell^{(\tau)}n-P>0$ (this corresponds to the estimate for the $Z$-distance).
\end{remark}

\begin{remark}
The complexity in \cref{prop:AG-Z-decoder} does not take into account the construction of the actual AG codes. This previous step can be carried out in polynomial time using standard algorithms for computing Riemann--Roch spaces, e.g., see \cite[Alg.~13 and Rem.~14]{hess_computing_RS_spaces}.
\end{remark}


\end{document}